\documentclass[manuscript]{acmart}

\usepackage{xpatch}
\usepackage{blindtext}
\makeatletter
\xpatchcmd{\ps@firstpagestyle}{Manuscript submitted to ACM}{}{\typeout{First patch succeeded}}{\typeout{first patch failed}}
\xpatchcmd{\ps@standardpagestyle}{Manuscript submitted to ACM}{}{\typeout{Second patch succeeded}}{\typeout{Second patch failed}}    \@ACM@manuscriptfalse
\makeatother
\setcopyright{none}
\renewcommand\footnotetextcopyrightpermission[1]{} 
\newif\iffull
\fulltrue

\usepackage{balance}
  
\usepackage{algorithm}
\usepackage[noend]{algpseudocode}
\usepackage{amssymb}

\newcommand{\id}{\ensuremath{\text{id}}}

\newcommand{\poly}{\text{poly}}

\def\E{{\textrm E}}
\def\N{\mathbb{N}}
\newcommand{\lf}{\lfloor}
\newcommand{\rf}{\rfloor}

\newtheorem{corollary}{Corollary}

\def\BibTeX{{\rm B\kern-.05em{\sc i\kern-.025em b}\kern-.08emT\kern-.1667em\lower.7ex\hbox{E}\kern-.125emX}}
    
\begin{document}

%
\title{Distributed Computation in Node-Capacitated Networks}
\iffull{\titlenote{This is an extended version of a paper that will appear at SPAA 2019.}}\fi

%

\author{John Augustine}
\email{augustine@iitm.ac.in}
\affiliation{%
  \institution{IIT Madras}
  \state{India}
}

\author{Mohsen Ghaffari}
\email{ghaffari@inf.ethz.ch}
\affiliation{%
  \institution{ETH Zurich}
  \state{Switzerland}
}

\author{Robert Gmyr}
\email{rgmyr@uh.edu}
\affiliation{%
  \institution{University of Houston}
  \state{USA}
}

\author{Kristian Hinnenthal}
\author{Christian Scheideler}
\email{{krijan, scheidel}@mail.upb.de}
\affiliation{%
  \institution{Paderborn University}
  \state{Germany}
}

\author{Fabian Kuhn}
\email{kuhn@cs.uni-freiburg.de}
\affiliation{%
  \institution{University of Freiburg}
  \state{Germany}
}

\author{Jason Li}
\email{jmli@cs.cmu.edu}
\affiliation{%
  \institution{Carnegie Mellon University}
  \state{USA}
}

%
\renewcommand{\shortauthors}{Augustine et al.}

%
\begin{abstract}
In this paper, we study distributed graph algorithms in networks in which the nodes have a limited communication capacity.
Many distributed systems are built on top of an underlying networking infrastructure, for example by using a virtual communication topology known as an \emph{overlay network}.
Although this underlying network might allow each node to directly communicate with a large number of other nodes, the amount of communication that a node can perform in a fixed amount of time is typically much more limited.

We introduce the \emph{Node-Capacitated Clique} model as an abstract communication model, which allows us to study the effect of nodes having limited communication capacity on the complexity of distributed graph computations.
In this model, the $n$ nodes of a network are connected as a clique and communicate in synchronous rounds.
In each round, every node can exchange messages of $O(\log n)$ bits with at most $O(\log n)$ other nodes.
When solving a graph problem, the input graph $G$ is defined on the same set of $n$ nodes, where each node knows which other nodes are its neighbors in $G$.

To initiate research on the Node-Capacitated Clique model, we present distributed algorithms for the \emph{Minimum Spanning Tree} (MST), \emph{BFS Tree}, \emph{Maximal Independent Set}, \emph{Maximal Matching}, and \emph{Vertex Coloring} problems.
We show that even with only $O(\log n)$ concurrent interactions per node, the MST problem can still be solved in polylogarithmic time.
In all other cases, the runtime of our algorithms depends linearly on the \emph{arboricity} of $G$, which is a constant for many important graph families such as planar graphs.
\end{abstract}

%
%


%

%
\maketitle

\section{Introduction}
Nowadays, most of the distributed systems and applications do not have a dedicated communication infrastructure, but instead share a common physical network with many others.
The logical network formed on top of this infrastructure is called an \emph{overlay network}.
For these systems, the amount of information that a node can send out in a single round does not scale linearly with the number of its incident edges.
Instead, it rather depends on the bandwidth of the connection of the node to the communication infrastructure as a whole.
For these networks, it is therefore more reasonable to impose a bound on the amount of information that a {\em node} can send and receive in one round, rather than imposing a bound on the amount of information that can be sent along each of its incident {\em edges}.
Also, the topology of the overlay network may change over time, and these changes are usually under the control of the distributed application.
To capture these aspects, we propose to study the so-called {\em Node-Capacitated Clique} model.
The model is inspired in part by the \emph{Congested Clique} model introduced first by Lotker, Patt-Shamir, Pavlov, and Peleg~\cite{LSPP05}, which has received significant attention recently~\iffull{\cite{JN18b, Len13, GP16, HPPS15, Kor16, LSPP05, JN18, CKKL15, DLP12, BKKL17, CPS17,HP14, HPS14, Leg16, konrad2018mis, ghaffari2018improved, ghaffari2017distributed, BMRT18, GN18}}\fi.

Similarly to the Congested Clique model, the nodes of the Node-Capacitated Clique are interconnected by a complete graph.
However, in the Node-Capacitated Clique every node can only send and receive at most $O(\log n)$ messages consisting of $O(\log n)$ bits in each round.
This limitation is added precisely to address the issue explained above.
It particularly rules out the possibility that the model allows one node to be in contact with up to $\Theta(n)$ other nodes at the same time; a property of the Congested Clique that seems to severely limit its practicability.
We comment that the capacity bound of $O(\log n)$ messages per node per round is a natural choice: it is small enough to ensure scalability and any smaller would require unnecessarily complicated techniques for the protocol to ensure nodes do not receive more messages than the capacity bound. 

Compared to traditional overlay network research, the Node-Capacitated Clique model has the advantage that it abstracts away the issue of designing and maintaining a suitable overlay network, for which many solutions have already been found in recent years.
Nevertheless, it is closely related to overlay networks:
every overlay network algorithm (i.e., an algorithm in which overlay edges can be established by introducing nodes to each other, and which satisfies the capacity bound of $O(\log n)$ messages) can be simulated in the Node-Capacitated Clique without any overhead.
Furthermore, any algorithm for our model can be simulated with a multiplicative $O(\log n)$ runtime overhead in the CRCW PRAM model (by assigning each processor $O(\log n)$ memory cells, and letting nodes write into randomly chosen cells of other processors), which in turn can be simulated with only $O(\log n)$ overhead by a network of constant degree \cite{Ran91}.
The Congested Clique model and its broadcast variant, on the other hand, are far more powerful (and arguably beyond what is possible in overlay networks):
Whereas in the Congested Clique a total of $\widetilde{\Theta}(n^2)$ bits can be transmitted in each round, in the Node-Capacitated Clique only $\widetilde{\Theta}(n)$ bits may be sent.
For example, the gossip problem---i.e., delivering one message from each node to every other node---can be solved in a single round in the Congested Clique, whereas the problem requires at least $\Omega(n/\log n)$ rounds in the Node-Capacitated Clique model.
Even the simple broadcast problem---i.e., delivering one message from one node to all nodes---already takes time $\Omega(\log n / \log \log n)$ in the Node-Capacitated Clique.

In this paper, we assume some edges of the network are marked as edges of an \emph{input graph} $G$, where each node knows which other nodes are its neighbors in $G$, and aim to solve \emph{graph problems} on $G$ using the power of the Node-Capacitated Clique.
Such edges can, for instance, be seen as edges of an underlying physical network, or represent relations between nodes in social networks.
Our results in that direction also turn out to be useful for some other theoretical models as well: they are relevant for \emph{hybrid networks} \cite{GHSS17} and also the \emph{$k$-machine model} for processing large scale graphs \cite{KNPR15}.

The concept of hybrid networks has just recently been considered in theory (e.g., \cite{GHSS17}).
In a hybrid network, nodes have different \emph{communication modes}:
We are given a network of cheap links of arbitrary topology that is not under the control of the nodes and may potentially be changing over time.
In addition to that, the nodes have the ability to build arbitrary overlay networks of costly links that are fully under the control of the nodes.
Cell phones, for example, can communicate in an ad-hoc fashion via their WiFi interfaces, which is for free but only has a limited range, and whose connections may change as people move.
Additionally, they may use their cellular infrastructure, which comes at a price, but remains fully under their control.
Although in the idealized setting this overlay network may form a clique, to save costs, the nodes might want to exchange only a small amount of messages of small size in each communication round.
This property is captured by the Node-Capacitated Clique.
The network of cheap links, on the other hand, can be seen as an input graph in the Node-Capacitated Clique for which the nodes want to solve a graph problem of interest.

Another interesting application of the Node-Capacitated Clique is the recently introduced \emph{$k$-machine model} \cite{KNPR15}, which was designed for the study of data center level distributed algorithms for large scale graph problems.
Here, a data center with $k$ servers is modeled as $k$ machines that are fully interconnected and capable of executing synchronous message passing algorithms.
A standard approach for the $k$-machine model is to partition the input graph in a fair way so that each machine stores a set of nodes of the input graph with their incident edges.
It is quite natural to simulate algorithms designed for the Node-Capacitated Clique model in the $k$-machine model.
Precisely, any algorithm that requires $T$ rounds in the Node-Capacitated Clique model can be simulated to take at most time $\widetilde{O}(nT/k^2)$.
\iffull{
The details of this simulation can be found in Appendix~\ref{app:simulation}.
} \fi
To illustrate the usefulness of this simulation, we remark that the running time of the fast minimum spanning tree algorithm provided by Pandurangan et al.~\cite{PRS16} can be obtained simply by converting the algorithm we provide in this work to the $k$-machine model.

As we demonstrate in this paper, many graph problems can be solved efficiently in the Node-Capacitated Clique, which shows that many interesting problems can be solved efficiently in distributed systems based on an overlay network over a shared infrastructure as well as hybrid networks and server systems.

\subsection{Model and Problem Statement}

In the \emph{Node-Capacitated Clique model} we consider a set $V$ of $n$ computation entities that we model as nodes of a graph.
Each node has a unique identifier consisting of $O(\log n)$ bits and every node knows the identifiers of all nodes such that, on a logical level, they form a complete graph.
Note that since every node knows the identifier of every other node, the nodes also know the total number of nodes $n$.
As node identifiers are common knowledge, without loss of generality we can assume that the identifiers are from the set $\{0, 1,  \dots, n-1 \}$.

The network operates in a synchronous manner with time measured in rounds.
In every round, each node can perform an arbitrary amount of local computation and send distinct messages consisting of $O(\log n)$ bits to up to $O(\log n)$ other nodes.
The messages are received at the beginning of the next round.
A node can receive up to $O(\log n)$ messages.
If more messages are sent to a node, it receives an \emph{arbitrary} subset of $O(\log n)$ messages.
Additional messages are simply dropped by the network.

Let $G$ be an undirected graph $G=(V, E)$ with an arbitrary edge set, but the same node set as the Node-Capacitated Clique.
We aim to solve graph problems on $G$ in the Node-Capacitated Clique model.
At the beginning, each node locally knows which identifiers correspond to its neighbors in $G$, but has no further knowledge about the graph.

\subsection{Related Work}

The Congested Clique model has already been studied extensively in the past years.
Problems studied in prior work include routing and sorting \cite{Len13}, minimum spanning trees \iffull{\cite{GP16,HPPS15,Kor16,LSPP05,JN18}}\fi, subgraph detection \iffull{\cite{CKKL15,DLP12,BMRT18}}\fi, shortest paths \cite{BKKL17,CKKL15}, local problems \cite{CPS17,HP14,HPS14}, minimum cuts \cite{JN18b, GN18}, and problems related to matrix multiplication \cite{CKKL15, Leg16}.
Some of the upper bounds are astonishingly small, such as the constant-time upper bound for routing and sorting and for the computation of a minimum spanning tree, demonstrating the power of the Congested Clique model.

While almost no non-trivial lower bounds exist for the Congested Clique model (due to their connection to circuit complexity \cite{DKO14}), various lower bounds have already been shown for the more general CONGEST model \iffull{\cite{DHKK11,FHW12,KP98,LP13,Nan14,PR00,Elk04}}\fi.
As pointed out in \cite{KS17}, the reductions used in these lower bounds usually boil down to constructing graphs with bottlenecks, that is, graphs where large amounts of information have to be transmitted over a small cut.
As this is not the case for the Node-Capacitated Clique, the lower bounds are of limited use here.
Therefore, it remains interesting to determine upper and lower bounds for the Node-Capacitated Clique.

Hybrid networks have only recently been studied in theory.
An example is the hybrid network model proposed in \cite{GHSS17}, which allows the design of much faster distributed algorithms for graph problems than with a classical communication network.
Also, the problem of finding short routing paths with the help of a hybrid network approach has been considered \cite{JKSS18}.
A priori, these papers do not assume that the nodes are completely interconnected, so extra measures have to be taken to build up appropriate overlays.
Abstracting from that problem, the Node-Capacitated Clique allows one to focus on how to efficiently exchange information in order to solve the given problems.

The graph problems considered in this paper have already been extensively studied in many different models.
In the CONGEST model, for example, a breadth-first search can trivially be performed in time $O(D)$.
There exists an abundance of algorithms to solve the maximal independent set, the maximal matching, and the coloring problem in the CONGEST model (see, e.g., \cite{BEP16} for a comprehensive overview).
Computing a minimum spanning tree has also been well studied in that model (see, e.g., \iffull{\cite{Elk04,Elk06,PR00,DHKK11}}\fi).
Whereas the running times of the above-mentioned algorithms depend on $D$ and additional polylogarithmic factors, there have also been proposed algorithms to solve such problems more efficiently in graphs with small \emph{arboricity} \cite{BE09, BE10, BE11, BEP16, KP11, KP12}.
Notably, Barenboim and Khazanov \cite{BK18} show how to solve a variety of graph problems in the Congested Clique efficiently given such graphs, e.g., compute an $O(a)$-orientation in time $O(\log a)$, an MIS in time $O(\sqrt{a})$, and an $O(a)$-coloring in time $O(a^\varepsilon)$, where $a$ is the arboricity of the given graph.
The algorithms make use of the \emph{Nash-Williams forest-decomposition} technique \cite{Nas64}, which is one of the key techniques used in our work.

\subsection{Our Contribution}

\begin{table}[t]
  \centering
  \renewcommand{\arraystretch}{1}
  \begin{tabular*}{\columnwidth}{@{}p{0.45\columnwidth} @{}p{0.4\columnwidth} @{}p{0.13\columnwidth} }
    \toprule
    Problem & Runtime & Section\\ \midrule
    Minimum Spanning Tree & $O(\log^4 n)$ & \ref{sec:mst}\\
    BFS Tree & $O((a + D + \log n)\log n)$ & \ref{sec:bfs} \\
    Maximal Independent Set & $O((a + \log n)\log n)$ & \ref{sec:mis} \\
    Maximal Matching & $O((a + \log n)\log n)$ & \ref{sec:matching} \\
    $O(a)$-Coloring &  $O((a + \log n)\log^{3/2}n)$ & \ref{sec:coloring} \\
    \bottomrule
  \end{tabular*}
  \caption{
    An overview of our results.
    We use $a$ for arboricity and $D$ to denote the diameter of the given graph.
  }
  \label{tab:ourContribution}
  \vspace{-20px}
\end{table}
We present a set of basic communication primitives and then show how they can be applied to solve certain graph problems (see Table~\ref{tab:ourContribution} for an overview).
Note that for many important graph families such as planar graphs, our algorithms have polylogarithmic runtime (except when depending on the diameter $D$).

Although many of our algorithms rely on existing algorithms from literature, we point out that most of these algorithms cannot be executed in the Node-Capacitated Clique in a straight-forward fashion.
The main reason for that is that high-degree nodes cannot efficiently communicate with all of their neighbors \emph{directly} in our model, which imposes significant difficulties to the application of the algorithms.
To overcome these difficulties, we present a set of basic tools that still allow for efficient communication, and combine it with variations of well-known algorithms and novel techniques.
Notably, we present an algorithm to compute an \emph{orientation} of the input graph $G$ with arboricity $a$, in which each edge gets assigned a direction, ensuring that the outdegree of any node is at most $O(a)$.
The algorithm is later used to efficiently construct \emph{multicast trees} to be used for communication between nodes.
Achieving this is a highly nontrivial task in our model and requires a combination of techniques, ranging from aggregation and multicasting to shared randomness and coding techniques.
We believe that many of the presented ideas might also be helpful for other applications in the Node-Capacitated Clique.

Although proving lower bounds for the presented problems seems to be a highly nontrivial task, we believe that many problems require a running time linear in the arboricity.
For the MIS problem, for example, it seems that we need to communicate at least $1$ bit of information about every edge (typically in order for a node of the edge to learn when the edge is removed from the graph because the other endpoint has joined the MIS).
However, explicitly proving such a lower bound in this model seems to require more than our current techniques in proving multi-party communication complexity lower bounds.

\section{Preliminaries}

In this section, we first give some basic definitions and describe a set of communication primitives needed throughout the paper.

\subsection{Basic Definitions and Notation}

Let $G = (V,E)$ be an undirected graph.
The \emph{neighborhood} of a node $u$ is defined as $N(u) = \{v \in V \mid \{u,v\} \in E$\}, and $d(u) = |N(u)|$ denotes its degree.
With $\Delta = \max_{u \in V}(d(u))$ we denote the maximum degree of all nodes in $G$, and $\overline{d} = \sum_{u \in V}d(u)/n$ is the average degree of all nodes.
The \emph{diameter} $D$ of $G$ is the maximum length of all shortest paths in $G$.

The \emph{arboricity} $a$ of $G$ is the minimum number of forests into which its edges can be partitioned.
Since the edges of any graph with maximum degree $\Delta$ can be greedily assigned to $\Delta$ forests, $a \le \Delta$.
Furthermore, since the average degree of a forest is at most $2$, and the edges of $G$ can be partitioned into $a$ forests, $\overline{d} \le 2a$.
Graphs of many important graph families have small arboricity although their maximum degree might be unbounded.
For example, a tree obviously has arboricity $1$.
Nash-Williams \cite{Nas64} showed that the arboricity of a graph $G$ is given by $\max_{H \subseteq G}(m_H/(n_H - 1))$, where $H \subseteq G$ is a subgraph of $G$ with at least two nodes and $n_H$ and $m_H$ denote the number of nodes and edges of $H$, respectively.
Therefore, any planar graph, which has at most $6n - 3$ edges, has arboricity at most $3$.
In fact, any graph with \emph{genus} $g$, which is the minimum number of handles that must be added to the plane to embed the graph without any crossings, has arboricity $O(\sqrt{g})$ \cite{BE10}.
Furthermore, it is known that the family of graphs that \emph{exclude a fixed minor} \cite{DL98} and the family of graphs with bounded \emph{treewidth} \cite{DW07} have bounded arboricity.

An \emph{orientation} of $G$ is an assignment of \emph{directions} to each edge, i.e., for every $\{u,v\} \in E$ either $u \rightarrow v$ ($u$ is directed to $v$) or $v \rightarrow u$ ($v$ is directed to $u$).
If $u \rightarrow v$, then $u$ is an \emph{in-neighbor} of $v$ and $v$ is an \emph{out-neighbor} of $u$.
For $u \in V$ define $N_{in}(u) = \{ v \in V \mid v \rightarrow u\}$ and $N_{out}(u) = \{ v \in V \mid u \rightarrow v\}$.
The \emph{indegree} of a node $u$ is defined as $d_{in}(u) = |N_{in}(u)|$ and its \emph{outdegree} is $d_{out}(u) = |N_{out}(u)|$.
A \emph{$k$-orientation} is an orientation with maximum outdegree $k$. For a graph with arboricity $a$, there always exists an $a$-orientation: we root each tree of every forest arbitrarily and direct every edge from child to parent node.

To allow each node to efficiently gather information sent to it by other nodes, our communication primitives make heavy use of \emph{aggregate functions}.
An aggregate function $f$ maps a multiset $S=\{x_1,\ldots,x_N\}$ of input values to some value $f(S)$.
For some functions $f$ it might be hard to compute $f(S)$ in a distributed fashion, so we will focus on so-called \emph{distributive aggregate functions}:
An aggregate function $f$ is called distributive if there is an aggregate function $g$ such that for any multiset $S$ and any partition $S_1,\ldots,S_\ell$ of $S$, $f(S)=g(f(S_1),\ldots,f(S_\ell))$.
Classical examples of distributive aggregate functions are MAX, MIN, and SUM.

\iffull{
Our algorithms make heavy use of randomized strategies.
To show that the correctness and runtime of the algorithms hold with high probability (w.h.p.)\footnote{We say an event holds \emph{with high probability}, if it holds with probability at least $1 - 1/n^c$ for any fixed constant $c > 0$.}, we use a generalization of the \emph{Chernoff bound} in \cite{SSS95} (Theorem 2):

\begin{lemma}\label{lem:generalChernoffBound}
  Let $X_1,\ldots,X_n$ be $k$-wise independent random variables with $X_i \in [0,b]$ and let $X = \sum_{i=1}^n X_i$.
  Then it holds for all $\delta \ge 1$, $\mu \ge E[X]$, and $k \ge \lceil \delta\mu \rceil$
  \[
    \Pr[X \ge (1+\delta)\mu] \le e^{-\min[\delta^2, \delta]\cdot \mu /(3b)}.
  \]
\end{lemma}

} \fi

\subsection{Communication Primitives} \label{sec:primitives}

Our algorithms make heavy use of a set of communication primitives, which are presented in this section.
Whereas the \emph{Aggregate-and-Broadcast algorithm} will be used as a general tool for aggregation and synchronization purposes, the other primitives are used to allow nodes to send and receive messages to and from specific sets of nodes associated with them.
Note that a node is not able to send or receive a large set of messages in few rounds; the center of a star, for example, would need linear time to deliver messages to all of its neighbors.
If, however, the number of \emph{distinct} messages a node has to send is small, or if messages destined at a node can be \emph{combined} using an aggregate function, then messages can be efficiently delivered using a randomized routing strategy.
Due to space limitations, we only present the high-level ideas of our algorithms and state their results.
\iffull{
The full description and all proofs can be found in Appendix~\ref{app:primitives}.
}
\fi

\paragraph*{Butterfly Simulation.}
To distribute local communication load over all nodes of the network, our algorithms rely on an emulation of a \emph{butterfly network}.
Formally, for $d\in \mathbb{N}$, the $d$-dimensional butterfly is a graph with node set $[d+1] \times [2^d]$, where we denote $[k] = \{0, \ldots, k-1\}$, and an edge set $E_1 \cup E_2$ with
\begin{align*}
  E_1 = &\{\{(i,\alpha),(i+1,\alpha)\} \mid i\in [d], \; \alpha\in [2^d] \}, \\
  E_2 = &\{\{(i,\alpha),(i+1,\beta)\} \mid i\in [d], \alpha,\beta \in [2^d], \\
        &\text{$\alpha$ and $\beta$ differ only at the $i$-th bit} \}.
\end{align*}
The node set $\{(i,j) \mid j \in [2^d]\}$ represents \emph{level $i$} of the butterfly, and node set $\{(i,j) \mid i \in [d+1]\}$ represents \emph{column $j$} of the butterfly.
In our algorithms, every node $u \in V$ with identifier $i \le 2^d - 1$ emulates the complete column $i$ of the $d$-dimensional butterfly with $d = \lfloor \log n \rfloor$.
Since $u$ knows the identifiers of all other nodes, it knows exactly which nodes emulate its neighbors in the butterfly.
As every node in the Node-Capacitated Clique can send and receive $O(\log n)$ messages in each round, and the butterfly is of constant degree, a communication round in the butterfly can be simulated in a single round in our model.

\paragraph*{Aggregate-and-Broadcast Problem.}
We are given a distributive aggregate function $f$ and a set $A \subseteq V$, where each member of $A$ stores exactly one input value.
The goal is to let every node learn $f(\text{inputs of } A)$.

\begin{theorem}\label{thm:aggregateAndBroadcast}
  There is an \emph{Aggregate-and-Broadcast Algorithm} that solves any Aggregation Problem in time $O(\log n)$.
\end{theorem}

In principle, the algorithm first aggregates all values from the topmost (i.e., level $0$) to the bottommost level (i.e., level $d$) of the butterfly, and then broadcasts the result upwards to all nodes in the butterfly.

\paragraph*{Aggregation Problem.}
We are given a distributive aggregate function $f$ and a set of \emph{aggregation groups} $\mathcal{A} = \{A_1, \ldots , A_N\}$, $A_i \subseteq V$, $i \in \{1, \ldots, N\}$ with targets $t_1,\ldots,t_N \in V$, where each node holds exactly one \emph{input value} $s_{u,i}$ for each aggregation group $A_i$ of which it is a \emph{member}, i.e., $u \in A_i$.\footnote{We only enumerate the aggregation groups from $1, \ldots, N$ to simplify the presentation of the algorithm.
Actually, we only require each aggregation group to be uniquely identified, which can easily be achieved for all algorithms in this paper.}
Note that a node may be member or target of multiple aggregation groups.
The goal is to aggregate these input values so that eventually $t_i$ knows $f(s_{u,i} \mid u \in A_i)$ for all $i$.
We define $L=\sum_{i=1}^N |A_i|$ to be the \emph{global load} of the Aggregation Problem, and the \emph{local load} $\ell = \ell_1 + \ell_2$, where $\ell_1 = \max_{u \in V} |\{ i \in \{1,\ldots,N\} \mid u \in A_i\}|$ and $\ell_2 = \max_{u \in V} |\{ i \in \{1,\ldots,N\} \mid u = t_i\}|$.
Whereas the global load captures the total number of messages that need to be processed, $\ell_1$ and $\ell_2$ indicate the work required for inserting messages into the butterfly, or sending aggregates from butterfly nodes to their targets, respectively.
We require that every node knows the identifier and target of all aggregation groups it is a member of, and an upper bound $\hat{\ell_2}$ on $\ell_2$.

\begin{theorem}\label{thm:aggregationProblem}
  There is an \emph{Aggregation Algorithm} that solves any Aggregation Problem in time $O(L/n + (\ell_1 + \hat{\ell_2})/\log n + \log n)$, w.h.p.\iffull{}\fi
\end{theorem}

From a very high level, the algorithm works as follows.
First, packets are sent to random nodes of the topmost level of the butterfly.
Then, packets belonging to the same aggregation group $A_i$ are routed to an intermediate target $h(i)$ in the bottommost level of the butterfly using a (pseudo-)random hash function $h$ and a variant of the random rank routing protocol \cite{Ale82, Upf82}.
Whenever two packets belonging to the same aggregation group collide on a butterfly node, they are combined using the function $f$.
Finally, the result of aggregation group $A_i$ is sent from its intermediate target to its actual target $t_i$.

The intermediate steps of the algorithm are synchronized using a variant of the Aggregate-and-Broadcast algorithm:
Every node delays its participation in an aggregation until having finished the current step.
Once the aggregation finishes, all nodes become informed about a common round to start the next step.
Termination of the routing protocol can easily be determined by passing down tokens in the butterfly.
We also use the same techniques to achieve synchronization for all other algorithms in this paper without explicitly mentioning it.

Note that common hash functions require \emph{shared randomness}.
Although in the remainder of this paper we assume that all hash functions behave like perfect random functions, it can be shown that it suffices to use $\Theta(\log n)$-wise independent hash functions (see, e.g., \cite{CRSW13} and the references therein):
Whenever we aim to show that the outcome of a random experiment deviates from the expected value by at most $O(\log n)$, w.h.p.,
\iffull{we can immediately use Lemma~\ref{lem:generalChernoffBound}; }
\fi
if the deviation we aim to show is higher, we can partition events in a suitable way so that we only need $\Theta(\log n)$-wise independence for each subset of events, and the sum of the deviations does not exceed the overall desired deviation.
To agree on such hash functions, all nodes have to learn $\Theta(\log^2 n)$ random bits.
This can be done by letting the node with identifier $0$ broadcast $\Theta(\log n)$ messages, each consisting of $\log n$ bits, to all other nodes using the butterfly.

\paragraph*{Multicast Tree Setup Problem.}
We are given a set of \emph{multicast groups} $\mathcal{A} = \{A_1,\ldots,A_N\}$, $A_i \subseteq V$, with sources $s_1, \ldots, s_N \in V$ such that each node is source of at most one multicast group (but possibly member of multiple groups).
The goal is to set up a \emph{multicast tree} $T_i$ in the butterfly for each $i \in \{1, \ldots, N\}$ with root $h(i)$, which is a node uniformly and independently chosen among the nodes of the bottommost level of the butterfly, and a unique and randomly chosen leaf $l(i,u)$ in the topmost level for each $u \in A_i$.
Let $L = \sum_{i=1}^N |A_i|$, $\ell = \max_{u \in V} |\{ i \in \{1, \ldots, N\} \mid u \in A_i \}|$ and define the \emph{congestion} of the multicast trees to be the maximum number of trees that share the same butterfly node.
We require that each node $u \in V$ knows the identifier and source of all multicast groups it is a member of.

\begin{theorem} \label{thm:treesetup}
  There is a \emph{Multicast Tree Setup Algorithm} that solves any Multicast Tree Setup Problem in time $O(L/n + \ell/\log n + \log n)$, w.h.p.
  The resulting multicast trees have congestion $O(L/n + \log n)$, w.h.p.
\end{theorem}

The algorithm shares many similarities with the Aggregation Algorithm; in fact, the multicast trees stem from the paths taken by the packets during an aggregation.
Alongside the aggregation, every butterfly node $u$ records for every $i \in \{1,\ldots,N\}$ all edges along which packets from group $A_i$ arrived during the routing towards $h(i)$, and declares them as edges of $T_i$.

\paragraph*{Multicast Problem.}
Assume we have constructed multicast trees for a set of multicast groups $\mathcal{A} = \{A_1,\ldots,A_N\}$, $A_i \subseteq V$, with sources $s_1, \ldots, s_N \in V$ such that each node is source of at most one multicast group.
The goal is to let every source $s_i$ send a message $p_i$ to all nodes $u \in A_i$.
Let $C$ be the congestion of the multicast trees and $\ell = \max_{u \in V} |\{ i \in \{1,\ldots,N\} \mid u \in A_i\}|$.
We require that the nodes know an upper bound $\hat{\ell}$ on $\ell$.

\begin{theorem}\label{thm:multicastProblem}
  There is a \emph{Multicast Algorithm} that solves any Multicast Problem in time $O(C + \hat{\ell}/\log n + \log n)$, w.h.p.
\end{theorem}

The algorithm multicasts messages by sending them upwards the multicast trees, performing our routing strategy in "reverse order".
We remark that similar to the Aggregation Algorithm, the Multicast Algorithm may easily be extended to allow a node to be source of multiple multicasts; however, we will only need the simplified variant in our paper.

\paragraph*{Multi-Aggregation Problem.}
We are given a set of multicast \\groups $\mathcal{A} = \{A_1, \ldots, A_N\}$, $A_i \subseteq V$, with sources $s_1,\ldots,s_N \in V$ such that every source $s_i$ stores a multicast packet $p_i$, and every node is source of at most one multicast group.
We assume that multicast trees for the multicast groups with congestion $C$ have already been set up.
The goal is to let every node $u \in V$ receive $f(\{p_i \mid u \in A_i\})$ for a given distributive aggregate function $f$.

\begin{theorem}\label{thm:multiaggregationProblem}
  There is a \emph{Multi-Aggregation Algorithm} that solves any Multi-Aggregation Problem in time $O(C + \log n)$, w.h.p.
\end{theorem}

The Multi-Aggregation algorithm combines all of the previous algorithms to allow a node to first multicast a message to a set of nodes associated with it, and then aggregate all messages destined at it.
More precisely, each source $s_i$ first multicasts its packet $p_i$ to all leaves in its multicast tree.
Every node $l(i,u)$ then \emph{maps} $p_i$ to a packet $(\id(u), p_i)$ for all $i$ and $u \in A_i$.
The resulting packets are randomly distributed among the nodes of the topmost level of the butterfly.
Finally, all packets associated with identifier $\id(u)$ for some $u$ are aggregated towards an intermediate target $h(\id(u))$ on level $d$ using the aggregate function $f$ as in the Aggregation Algorithm.
From there, the result is finally delivered to $u$.
For applications beyond our paper, the algorithm may also be extended to allow nodes to be source of multiple multicast groups, and to receive aggregates corresponding to distinct aggregations.

\section{Minimum Spanning Tree} \label{sec:mst}
As a first example of graph algorithms for the Node-Capacitated Clique, we describe an algorithm that computes a \emph{minimum spanning tree (MST)} in time $O(\log^4 n)$.
More specifically, for every edge in the input graph $G$, one of its endpoints eventually knows whether the edge is in the MST or not.
We assume that each edge of $G$ has an integral weight in $\{1,2,\ldots,W\}$ for some positive integer $W = \poly(n)$.

\paragraph*{High-Level Description.}
From a high level, our algorithm mimics Boruvka's algorithm with Heads/Tails clustering, which works as follows.
Start with every node as its own component.
For $O(\log n)$ iterations, every component $C$ (1) finds its \emph{lightest}, i.e., minimum-weight, edge out of the component that connects to the other components, (2) flips a Heads/Tails coin, and (3) learns the coin flip of the component $C'$ on the other side of the lightest edge.
If $C$ flips Tails and $C'$ flips Heads, then the edge connecting $C$ to $C'$ is added to the MST, and thus effectively component $C$ merges with component $C'$ (and whatever other components that are merging with $C'$ simultaneously).
It is well known that, w.h.p., all nodes get merged into one component within $O(\log n)$ iterations and the added edges form an MST (see, e.g., \cite{GhaffariHaeupler16, ghaffariKS17}).

\paragraph*{Details of the Algorithm.}
Over the course of the algorithm, each component $C \subseteq V$ maintains a \emph{leader node} $l(C)\in C$ whose identifier is known to every node in the component.
Furthermore, we maintain a multicast tree for each component $C$ with source $l(C)$ and corresponding multicast group $C \setminus \{l(C)\}$.
We ensure that the set of multicast trees has congestion $O(\log n)$.
In each round of Boruvka's algorithm with the partition of $V$ into components $\{C_1, \ldots, C_N\}$, every leader $l(C_i)$ flips Heads/Tails and multicasts the result to all nodes in its component by using the Multicast Algorithm of Theorem~\ref{thm:multicastProblem}.
As the multicast trees have congestion $O(\log n)$, and $\hat{\ell} = 1$ as every node is in exactly one component, this takes time $O(\log n)$, w.h.p.

For each component $C$, the leader then learns the lightest edge to a neighbor in $V \setminus C$ in time $O(\log^2 n \log W)$.
This is a highly nontrivial task that we address later.
Afterwards, the leader multicasts the lightest edge $\{u, v\} \in (C \times (V\setminus C)) \cap E$ to every node in its component, which can again be done in time $O(\log n)$.
For each component $C$ that flips Tails, the node $u \in C$ incident to the lightest outgoing edge $\{u, v\}$ now has to learn whether $v$'s component $C'$ has flipped Heads, and, if so, the identifier of $l(C')$.
Therefor, $u$ \emph{joins} a multicast group $A_{\id(v)}$ with source $v$, i.e., declares itself a member of $A_{\id(v)}$ and constructs multicast trees with the help of Theorem~\ref{thm:treesetup}.
As every node is member of at most one multicast group, setting up the corresponding trees with congestion $O(\log n)$ takes time $O(\log n)$, w.h.p.
By using the Multicast Algorithm, the endpoints of all lightest edges learn the result of the coin flip and the identifier of their adjacent component's leader in time $O(\log n)$.

If for the edge $\{u, v\}$ the component $C'$ of $v$ has flipped Heads, then $u$ sends the identifier of the leader of $C'$ to its own leader, which in turn informs all nodes of $C$ using a multicast.
Note that thereby only $u$ learns that $\{u, v\}$ is an edge of the MST, but not $v$.
Finally, the multicast trees of the resulting components are rebuilt by letting each node join a multicast group corresponding to its new leader.
As the components are disjoint, the resulting trees with congestion $O(\log n)$ are built in time $O(\log n)$, w.h.p.

\paragraph*{Finding the Lightest Edge.}
To find the lightest edge of a component, we ``sketch'' its incident edges.
Our algorithm follows the procedure $\texttt{FindMin}$ of~\cite{KKT15}, with the ``broadcast-and-echo'' subroutine inside each component replaced by multicasts and aggregations (i.e., executions of the Multicast and Aggregation Algorithm) from/to the leader to/from the entire component.
As argued above, and due to Theorem~\ref{thm:aggregationProblem}, both steps can be performed in time $O(\log n)$, w.h.p.
We highlight the main steps of $\texttt{FindMin}$, and refer the reader to~\cite{KKT15} for the details and proof.

Initially, we bidirect each edge into two arcs in opposite directions, and define the identifier $\id(u,v) = \id(u) \circ \id(v)$, where $\circ$ denotes the concatenation of two binary strings.
We will apply binary search to the weights of edges so that we can find the lightest outgoing edge.
Every iteration has a current range $[L,R ]\subseteq [1,W]$ such that the lightest edge out has weight in that range.
To compute the next range, the algorithm determines whether there is an edge out of $[L,M]$, where $M:=\lf(L+R)/2\rf$.
If so, the new range becomes $[L,M]$; otherwise, the new range is $[M+1,R]$.\footnote{The algorithm \texttt{FindMin} of~\cite{KKT15} actually uses a ``$\Theta(\log n)$-ary'' search instead of binary search, but we replace it with binary search here for simplicity of explanation.}
The remaining task is to solve the following subproblem: given a range $[a,b]$, determine whether there exists an outgoing edge with weight in $[a,b]$.

To sketch their incident edges, the nodes use a (pseudo-)random hash function $h$ that maps each edge identifier to $\{0,1\}$.
For a node $u$, define
\[
  h^\uparrow(u):=\sum_{\substack{v\in N(u): \\ w(u,v)\in[a,b]}}h(\id(u,v))\mod2,
\]
and
\[
  h^\downarrow(u):=\sum_{\substack{v\in N(u): \\ w(u,v)\in[a,b]}}h(\id(v,u))\mod2,
\]
and for component $C \subseteq V$, define $h^\uparrow(C):=\sum_{u\in C}h^\uparrow(u)$ and $h^\downarrow(C)$ similarly.
Observe that the unordered sets $\{\id(u,v):u\in C,v\in N(u),w(u,v)\in[a,b]\}$ and $\{\id(v,u):u\in C,v\in N(u),w(u,v)\in[a,b]\}$ are the same if and only if component $C$ does not have an outgoing edge with weight in the range $[a,b]$.
Also, the hash function $h$ satisfies the property that, if two sets $S_1,S_2$ of integers are not equal, then the values of $\sum_{x\in S_1}h(x)\mod2$ and $\sum_{x\in S_2}h(x)\mod2$ are not equal with constant probability.
To compute the values of $h^\uparrow(C)$ and $h^\downarrow(C)$, each node $u\in C$ computes $h^\uparrow(u)$ and $h^\downarrow(u)$, and an aggregation towards the leader node is performed in each component $C$ with addition mod 2 as the aggregate function.
We can repeat this procedure $O(\log n)$ times so that w.h.p., there is no outgoing edge out of $C$ with weight in $[a,b]$ if and only if $h^\uparrow(C)$ and $h^\downarrow(C)$ are equal in every trial.
Note that this requires the nodes to know $O(\log n)$ different hash functions; by the discussion in Section~\ref{sec:primitives}, the necessary $O(\log^3 n)$ bits can be retrieved beforehand in $O(\log^2 n)$ rounds.

The running time analysis from~\cite{KKT15}, modified to count the number of ``broadcast-and-echo'' subroutines, can be rewritten as follows.

\begin{lemma}[\cite{KKT15}, Lemma 2]\label{lemma:Sketch}
  The leader node of each component learns the lightest edge out of its component within $O(\log W\log n)$ iterations of multicasts and aggregations, w.h.p.
\end{lemma}

Since each iteration can be performed in time $O(\log n)$, and there are $O(\log n)$ phases of Boruvka's algorithm, w.h.p., we conclude the following theorem.

\begin{theorem}
  The algorithm computes an MST in time $O(\log^4 n)$, w.h.p.
\end{theorem}

\section{Computing an $O(a)$-Orientation}

One of the reasons the MST problem can be solved very efficiently is because we only require \emph{one} endpoint of each edge to learn whether the edge is in the MST or not; otherwise, the problem seems to become significantly harder, as every node would have to learn some information about each incident edge.
We observe this difficulty for the other graph problems considered in this paper as well.
To approach this issue, we aim to set up multicast trees connecting each node with \emph{all} of its neighbors in $G$, allowing us to essentially simulate variants of classical algorithms.
As we will see, such trees can be set up efficiently if $G$ has small arboricity by first computing an $O(a)$-orientation of $G$, which is described in this section.

We present the \emph{Orientation Algorithm}, which computes an $O(a)$-orientation in time $O((a + \log n) \log n)$, w.h.p.
More specifically, the goal is to let every node learn a direction of all of its incident edges in $G$.
The algorithm essentially constructs a \emph{Nash-Williams forest-decomposition} \cite{Nas64} using the approach of \cite{BE10}.
From a high-level perspective, the algorithm repeatedly identifies low-degree nodes and removes them from the graph until the graph is empty.
Whenever a node leaves, all of its incident edges are directed away from it.
More precisely, the algorithm proceeds in phases $1, \ldots, T$.
Let $d_i(u)$ be the number of incident edges of a node $u$ that have not yet been assigned a direction at the beginning of phase $i$.
Define $\overline{d_i}$ to be the average degree of all nodes $u$ with $d_i(u) > 0$, i.e., $\overline{d_i} = \sum_{u \in V} d_i(u) / |\{ u \in V \mid d_i(u) > 0\}|$.
In phase $i$, a node $u$ is called \emph{inactive} if $d_i(u) = 0$, \emph{active} if $d_i(u) \le 2\overline{d_i}$, and \emph{waiting} if $d_i(u) > 2\overline{d_i}$.
In each phase, an edge $\{u,v\}$ gets directed from $u$ to $v$, if $u$ is active and $v$ is waiting, or if both nodes are active and $\id(u) < \id(v)$.
Thereby, each node is waiting until it becomes active in some phase, and remains inactive for all subsequent phases.
This results in a partition of the nodes into \emph{levels} $L_1, \ldots, L_T$, where level $i$ is the set $L_i$ of active nodes in phase $i$.
The lemma below follows from the fact that in every phase, at least half of all nodes that are not yet inactive become inactive, which can easily be shown, and that $\overline{d_i} \le 2a$, since any subgraph of $G$ can be partitioned into $a$ forests, whose average degree is at most $2$.

\begin{lemma}\label{lem:orientationHighLevel}
  The Orientation Algorithm takes $O(\log n)$ phases to compute an $O(a)$-orientation.
\end{lemma}


\subsection{Identification Problem}

It remains to show how a single phase can be performed efficiently in our model.
Here, the main difficulty lies in having active nodes determine which of their neighbors are already inactive in order to conclude the orientations of incident edges.
We approach this problem by solving the following \emph{Identification Problem}:
We are given a set $L \subseteq V$ of \emph{learning} nodes and a set $P \subseteq V$ of \emph{playing} nodes.
Every playing node knows a subset of its neighbors that are \emph{potentially} learning, i.e., it knows that none of the other neighbors are learning.
The goal is to let every learning node determine which of its neighbors are playing.

To solve such a problem, we present the \emph{Identification Algorithm}, which will later be used as a subroutine.
In this subsection, we represent each edge $\{u,v\}$ by two directed edges $(u,v)$ and $(v,u)$.
We assume that all nodes know $s$ (pseudo-)random hash functions $h_1, \ldots, h_s: E \rightarrow [q]$ for some parameters $s$ and $q$.
The hash functions are used to map every directed edge to $s$ \emph{trials}.
We say an edge $e$ \emph{participates} in trial $i$ if $h_j(e) = i$ for some $j$.

Let $u \in L$.
We refer to an edge $(u,v)$ as a \emph{red edge} of $u$, if $v$ is not playing, and a \emph{blue edge} of $u$, if $v$ is playing.
We identify each edge $(u,v)$ by the identifiers of its endpoints, i.e., $\id(u,v) = \id(u) \circ \id(v)$.
Let $X(i)$ be the XOR of the identifiers of all edges $(u,v)$ that participate in trial $i$, and $X'(i)$ be the XOR of the identifiers of all blue edges $(u,v)$ that participate in trial $i$.
Furthermore, let $x(i)$ be the total number of edges adjacent to $u$ that participate in trial $i$, and let $x'(i)$ be the number of blue edges that participate in trial $i$.

Our idea is to let $u$ use these values to identify all of its red edges; then it can conclude which of its neighbors must be playing.
Before describing this, we explain how the values are determined.
Clearly, the values $X(i)$ and $x(i)$ can be computed by $u$ by itself for all $i$.
The other values are more difficult to obtain as $u$ does not know which of its edges are blue.
To compute these values, we use the Aggregation Algorithm:
Each playing node $v$ is in aggregation group $A_{\id(w) \circ i}$ for every potentially learning neighbor $w$ and every trial $i$ such that
$(w,v)$ participates in trial $i$.
The input of $v$ for the group $A_{\id(w) \circ i}$ is $(\id(w,v), 1)$, where the first coordinate is used to let $w$ compute $X'(i)$, and the second coordinate is used to compute $x'(i)$.
Correspondingly, the aggregate function $f$ combines two inputs corresponding to the same aggregation group by taking the XOR of the first coordinate and the sum of the second coordinate.
Thereby, $u$ eventually receives both $X'(i)$ and $x'(i)$.

We now show how $u$ can identify its red edges using the aggregated information.
First, it determines a trial $i$ for which $x(i) = x'(i) + 1$.
Since neighbors that are not playing did not participate in the aggregation, in this case there is exactly one red edge $(u,v)$ such that $\id(u,v)$ is included in $X(i)$ but not in $X'(i)$.
Therefore, $\id(u,v)$ can be retrieved by taking the XOR of both values.
Having identified $\id(u,v)$, $u$ determines all trials in which $(u,v)$ participates using the common hash functions and "removes" $\id(u,v)$ from $X(i)$ by again computing the XOR of both.
It then decreases $x(i)$ by 1 and repeats the above procedure until no further edge can be identified.
If $u$ always finds a trial $i$ for which $x(i) = x'(i) + 1$, then it eventually has identified all red edges.
Clearly, all the remaining neighbors must be playing.

\begin{lemma}\label{lem:collision}
  Let $u \in L$ and assume that $u$ is incident to at most $p$ red edges.
  Let $s$ be the number of hash functions, and $q$ be the number of trials.\\
  \[
    \Pr[u \text{ fails to identify at least } k \text{ red edges}] \le 2\left(\frac{2sk}{q}\right)^{(s-2)k/2}
  \]
  for $q \ge 4esp$ and $s \ge 4$.
\end{lemma}

\begin{proof}
  $u$ fails to identify at least $k$ red edges if at some iteration of the above procedure there are $j \ge k$ edges left such that all edges participate only in trials in which at least two of the $j$ edges participate.
  Here, the $j$ edges participate in at most $\lfloor s \cdot j / 2 \rfloor$ many different trials, since otherwise there must be a trial in which only one edge participates.
  Therefore, the probability for that event is
  \begin{align*}
    \Pr &\le \sum_{j = k}^{p} \binom{p}{j} \binom{q}{sj / 2} \left(\frac{sj/2}{q}\right)^{sj} \\
    &\le \sum_{j = k}^{p} \left(\frac{ep}{j}\right)^{j} \left(\frac{2eq}{sj}\right)^{sj/2} \left(\frac{sj}{2q}\right)^{sj} \\
    &= \sum_{j = k}^{p} \left[ \left( \frac{ep}{j} \cdot \frac{sj}{2q} \right) \left( \frac{2eq}{sj} \cdot \frac{sj}{2q} \right)^{s/2} \cdot \left( \frac{sj}{2q} \right)^{s/2 - 1} \right]^j \\
    &= \sum_{j = k}^{p} \left[ \frac{e^2ps}{2q} \cdot \left( \frac{esj}{2q}\right)^{s/2 - 1} \right]^j \\
    &\le \sum_{j = k}^{p} \left(\frac{2sj}{q}\right)^{(s-2)j/2}
    \le 2\left(\frac{2sk}{q}\right)^{(s-2)k/2},
  \end{align*}
  where the last inequality holds because
  \iffull{
    \begin{align*}
        &\quad \left(\frac{2s(j+1)}{q}\right)^{(s-2)(j+1)/2} \\
        &= \left(\frac{2s(j+1)}{q}\right)^{(s-2)/2} \left(\frac{2s(j+1)}{q}\right)^{(s-2)j/2} \\
        &= \left(\frac{2s(j+1)}{q}\right)^{(s-2)/2} \left(\left(\frac{j+1}{j}\right)^j\right)^{(s-2)/2} \left(\frac{2sj}{q}\right)^{(s-2)j/2} \\
        &\le \left(\frac{2es(j+1)}{q}\right)^{(s-2)/2} \left(\frac{2sj}{q}\right)^{(s-2)j/2} \\
        &\le 1/2 \left(\frac{2sj}{q}\right)^{(s-2)j/2}. \qedhere
    \end{align*}
  }
  \fi
\end{proof}

\subsection{Details of the Algorithm}

Finally, we show how the Identification Algorithm can be used to efficiently realize a phase of the high-level algorithm in time $O(a + \log n)$, w.h.p.
In our algorithm every node learns the direction of all its incident edges in the phase in which it is active; however, its neighbors might learn their direction only in subsequent phases.
Each phase is divided into three \emph{stages}:
In Stage 1, every node determines whether it is active in this phase.
In Stage 2, every active node learns which of its neighbors are inactive.
Finally, in Stage 3 every active node learns which of its remaining neighbors, which must be either active or waiting, are active.
From this information, and since every node knows the identifiers of all of its neighbors, every active node concludes the direction of each of its incident edges.
In the following we describe the three stages of a phase $i$ in detail.

\paragraph*{Stage 1: Determine Active Nodes.}
We assume that all nodes start the stage in the same round.
First, every node $u$ that is not inactive needs to compute $d_i(u)$ (i.e., $d(u)$ minus the number of inactive neighbors) to determine whether it remains waiting or becomes active in this phase.
This value can easily be computed using the Aggregation Algorithm:
Every inactive node $v$, which already knows the orientation of each of its incident edges, is a member of every aggregation group $A_{\id(w)}$ such that $v \rightarrow w$.
As the input value of each node we choose $1$, the aggregate function $f$ is the sum, and $\ell_2 \le 1 =: \hat{\ell_2}$.
By performing the Aggregation Algorithm, $u$ determines the number of inactive neighbors, and, by subtracting the value from $d(u)$, computes $d_i(u)$.
Afterwards, the nodes use the Aggregate-and-Broadcast Algorithm to compute $\overline{d_i}$ and to achieve synchronization.

\paragraph*{Stage 2: Identify Inactive Neighbors.}
The goal of this stage is to let every active node learn which of its neighbors are inactive.
The stage is divided into two \emph{steps}:
In the first step, a large fraction of active nodes \emph{succeeds} in the identification of inactive neighbors.
The purpose of the second step is to take care of the nodes that were \emph{unsuccessful} in the first step, i.e., that only identified some, but not all, of their incident red edges.
In both steps we use the Identification Algorithm described in the previous section, and carefully choose the parameters to achieve that each step only takes time $O(a + \log n)$.

At the beginning of the first step, the nodes compute $d_i^* = \max_{u \in L_i}(d_i(u))$ by performing the Aggregate-and-Broadcast Algorithm.
Let $d^* = \max_{j \le i}d_i^*$, which is a value known to all nodes, and note that $d^* = O(a)$.
Then, the nodes perform the Identification Algorithm, where the active nodes are \emph{learning} and the inactive nodes are \emph{playing}.
Hence, the endpoints of the \emph{red edges} learned by the active nodes must either be active or waiting.
If we chose $s = c\log n$ and $q = 4ecd^*\log n$ for some constant $c > 6$ as parameters, then by Lemma~\ref{lem:collision} all nodes would learn all of their red edges, w.h.p., already in this step; however, this would take time $O(a\log n)$.
To reduce this to $O(a + \log n)$, we instead choose $s = c$ and $q = 4ecd^*\log n$ for some constant $c > 6$, and accept that nodes fail to identify some of their red edges in this step.
However, for this choice Lemma~\ref{lem:collision} implies that each node fails to identify at most $\log n$ red edges, w.h.p.

We now describe how these remaining edges are identified in the second step.
Let $U = \{u \in V \mid u \text{ is unsuccessful}\}$.
We divide $U$ into sets of \emph{high-degree} nodes $U_{high} = \{u \in U \mid (d(u) - d_i(u)) > n/\log n\}$ and of \emph{low-degree} nodes $U_{low} = \{u \in U \mid (d(u) - d_i(u)) \le n/\log n\}$ and consider the nodes of each set separately.
By dealing with high-degree nodes separately, we ensure that the global load required to let low-degree nodes identify their red edges reduces by a $\log n$ factor.
First, the nodes of $U_{high}$ (of which there are only $O(a + \log n)$, w.h.p.) broadcast their identifiers by using a variant of the Aggregate-and-Broadcast Algorithm:
Using the path system of the butterfly, every node $u \in U_{high}$ sends its identifier to the node $v$ with identifier $0$; however, messages are not combined.
Instead, whenever multiple identifiers contend to use the same edge in the same round, the smallest identifier is sent first.
After $v$ has received all identifiers, it broadcasts them in a pipelined fashion, i.e., one after the other, to all other nodes.
For every node $u \in A := \{u \in V \mid u \text { is active or waiting}\}$ define $R_u = U_{high} \cap N(u)$, i.e., $(v, u)$ is a red edge of $v$ for all $v \in R_u$.
Let $u \in A$.
For each $v \in R_u$, $u$ chooses a round from $\{1, \ldots, \max\{|R_u|, d_i^*\}\}$ uniformly and independently at random and sends its own identifier to $v$ in that round.
Afterwards, every high-degree node can identify all of its red edges.
As $\max_{u \in A}\{|R_u|, d_i^*\} = O(a + \log n)$, this takes time $O(a + \log n)$, w.h.p.

To let the low-degree nodes identify their red edges, we again use the Identification Algorithm.
First, in order to narrow down its set of potentially learning neighbors, every inactive node determines which of its neighbors are unsuccessful low-degree nodes.
Therefore, we let every inactive node $u$ join multicast group $A_{\id(v)}$ for all $u \rightarrow v$ such that $v$ is not inactive (recall that every inactive node knows the directions of all of its incident edges, and whether the other endpoint of each edge is inactive or not).
Every node $v \in U_{low}$ then informs its inactive neighbors by using the Multicast Algorithm.
Since every node is member of at most $d^*$ multicast groups, which is a value known to all nodes, the nodes know an upper bound on $\ell$ as required by the algorithm.
Having narrowed down the set of learning nodes and the sets of potentially learning neighbors to the unsuccessful ones only, the Identification Algorithm is performed once again.
As the parameters of the algorithm we choose $s = c\log n$ and $q = 4ec\log^2 n$ for some constant $c > 6$.

\paragraph*{Stage 3: Identify Active Neighbors.}
Finally, every active node has to learn which of the endpoints of its red edges are active.
In the following, let $\id(e) = \id(u) \circ \id(v)$ be the identifier of an edge given by its endpoints $u$ and $v$ such that $\id(u) < \id(v)$.
The nodes use two (pseudo-)random hash-function $h$, $r$, where $h$ maps the identifier of an edge $e$ to a node $h(\id(e)) \in V$ uniformly and independently at random, and $r$ maps its identifier to a round $r(\id(e)) \in \{1, \ldots, d_i^*\}$ uniformly and independently at random.
Every active node $u$ sends an \emph{edge-message} containing $\id(e)$ to $h(\id(e))$ in round $r(\id(e))$ for every incident edge $e$ leading to an active or waiting node.
Using this strategy, two adjacent active nodes $u$, $v$ send an edge-message containing $\id(\{u,v\})$ to the same node in the same round.
Whenever a node receives two edge-messages with the same edge identifier, it immediately responds to the corresponding nodes, which thereby learn that both endpoints are active.
\iffull{}\fi

\iffull{

\subsection{Analysis}

We now turn to the analysis of the algorithm. 
We mainly show the following lemma:

\begin{lemma}\label{lem:orientationPhase}
  In phase $i$ of the algorithm, every node $v \in L_i$ learns the directions of its incident edges.
  Each phase takes time $O(a + \log n)$, w.h.p.
  In every round, each node sends and receives at most $O(\log n)$ messages, w.h.p.
\end{lemma}

We present the proof in three parts:
first, we show the correctness of the algorithm, then analyze its runtime, and finally show that every node receives at most $O(\log n)$ messages in each round.

\begin{lemma}\label{lem:step1fail}
  In the first step, every active node fails to identify at most $\log n$ red edges, w.h.p.
\end{lemma}

\begin{proof}
  Note that every active node can only be adjacent to at most $p \le d^*$ active or waiting nodes, i.e., it is incident to at most $p$ red edges.
  Therefore, by Lemma~\ref{lem:collision}, the probability that an active node $u$ fails to identify at least $\log n$ red edges is
  \[
    2\left(\frac{2c\log n}{4ecd^*\log n}\right)^{(c-2)\log n/2}
    \le \frac{1}{2^{(c/2 - 1)\log n - 1}}
    \le \frac{1}{n^{c/2 - 2}}.
  \]
  Taking the union bound over all nodes implies the lemma.
\end{proof}

\begin{lemma}\label{lem:step2}
  After the second step, every active node has identified all of its red edges, w.h.p.
\end{lemma}

\begin{proof}
  If $u \in U_{high}$, then after having received the identifiers of all neighbors that are active or waiting, $u$ immediately knows its red edges.
  Now let $u \in U_{low}$.
  Since by Lemma~\ref{lem:step1fail} $u$ has at most $p \le \log n$ remaining red edges, by Lemma~\ref{lem:collision} we have that the probability that $u$ fails to identify at most one of its remaining red edges is at most
  \[
    2\left(\frac{2c \log n}{4ec\log^2 n}\right)^{(c\log n - 2) / 2} \le \frac{1}{2^{c\log n/2 - 2}} \le \frac{1}{n^{c/2 - 2}}.
  \]
  Taking the union bound over all nodes implies the lemma.
\end{proof}

To bound the runtime of the complete algorithm, we now prove that each stage takes time $O(a + \log n)$, w.h.p.

\begin{lemma}\label{lem:stage1}
  Stage 1 takes time $O(a + \log n)$, w.h.p.
\end{lemma}

\begin{proof}
  In the execution of the Aggregation Algorithm, every inactive node is member of at most $O(a)$ aggregation groups and every active node is target of at most one aggregation, i.e., $L = O(na)$ and $\ell_1 + \hat{\ell_2} = O(a)$.
  The lemma follows from Theorem~\ref{thm:aggregationProblem}.
\end{proof}

For the runtime of Stage 2 we need the following two lemmas.

\begin{lemma}\label{lem:highdeg}
  $|U_{high}| = O(a + \log n)$, w.h.p.
\end{lemma}

\begin{proof}
  Let $A = \{ u \in L_i \mid (d(u) - d_i(u)) > n/\log n\}$.
  Note that since $\overline{d} \le 2a$, we have that $\sum_{u \in V} d(u) \le 2an$, and therefore $|A| \le 2a\log n$.
  For $u \in A$ let $X_u$ be the binary random variable that is $1$, if $u$ is unsuccessful in the first step, and $0$, otherwise.
  By Lemma~\ref{lem:collision} and since $c > 6$, we have
  \[
    \Pr[X_u = 1] \le \frac{1}{\log^{c/2 - 2} n} \le \frac{1}{\log n}.
  \]
  Let $X = \sum_{u \in A} X_u$.
  $X$ is the sum of independent binary random variables with expected value $\E[X] \le 2a\log n/ \log n = 2a =: \mu$.
  Let $\delta = \max\{\alpha \log n / \mu, 1\}$ for some constant $\alpha > 3$, then by using the Chernoff bound of Lemma~\ref{lem:generalChernoffBound} we get that
  \[
    \Pr[X \ge (1 + \delta)\mu] \le e^{-\alpha \log n/3} \le \frac{1}{n^{\alpha/3}},
  \]
  and thus $X = O(a + \log n)$, w.h.p.
\end{proof}

\begin{lemma}\label{lem:lowdeg}
  $\sum_{u \in U_{low}} (d(u) - d_i(u)) = O(an/\log n + n)$, w.h.p.
\end{lemma}

\begin{proof}
  Let $A = \{ u \in L_i \mid (d(u) - d_i(u)) > n/\log n\}$.
  For a node $u \in A$, let $X_u$ be the random variable that is $d_u$, if $u$ is unsuccessful in the first step, and $0$, otherwise.
  From the proof of Lemma~\ref{lem:highdeg}, we have that $Pr[X_u = d_u] \le 1/\log n$.
  Let $A$ be the set of active nodes.
  Then $X = \sum_{u \in A}X_u$ is a sum of independent random variables with expected value $E[X] \le \sum_{u \in A}d(u) / \log n \le an / \log n =: \mu$.
  Note that $d(u) \le n/\log n$ for all $u \in A$.
  Therefore, we can use the Chernoff bound of Lemma~\ref{lem:generalChernoffBound} with $\delta = \max\{\alpha n / \mu, 1\}$ for some constant $\alpha>3$, and get
  \[
    \Pr[X \ge (1 + \delta) \mu] \le e^{- \alpha n \log n/(n3)} \le \frac{1}{n^{\alpha/3}}.
  \]
  Therefore, we have that $X = O(an/\log n + n)$, w.h.p.
\end{proof}

We are now ready to bound the runtime of Stage 2.

\begin{lemma}\label{lem:stage2}
  Stage 2 takes time $O(a + \log n)$, w.h.p.
\end{lemma}

\begin{proof}
  The computation of $d^*$ at the beginning of the first step takes time $O(\log n)$.
  To perform the first execution of the Identification Algorithm, every node has to learn $s = O(1)$ hash functions, which can be done in time $O(\log n)$ (see Section \ref{sec:primitives}).
  In the first execution of the Identification Algorithm, every active node $u$ is target of aggregation group $A_{\id(u) \circ i}$ for every trial $i$, and every inactive neighbor $v$ of $u$ is member of all aggregation groups $A_{\id(u) \circ i}$ such that $(u,v)$ participates in trial $i$.
  Therefore, every active node is target of at most $4ecd^* \log n$ and every inactive node is a member of at most $cd^*$ aggregation groups.
  Since both values are known to every node, the nodes know an upper bound $\hat{\ell_2} = 4ecd^* \log n$ on $\ell_2$.
  Since every inactive node is a member of at most $cd^*$ aggregation groups, the global load $L$ is bounded by $ncd^*$.
  By Theorem~\ref{thm:aggregationProblem}, the Aggregation Algorithm takes time
  \[
    O\left(\frac{ncd^*}{n} + \frac{4ecd^*\log n}{\log n} + \log n\right) = O(a + \log n),
  \]
  w.h.p., to solve the problem.

  Now consider the second step.
  By Lemma~\ref{lem:highdeg}, $|U_{high}| = O(a + \log n)$, w.h.p..
  A simple delay sequence argument can be used to show that all identifiers are broadcasted within time $O(a + \log n)$.
  Informing each node in $U_{high}$ about its red edges takes an additional $O(a + \log n)$ rounds, as $|R_u| = O(a + \log n)$ for every node $u$ and $d_i^* = O(a)$.

  The multicast trees to handle low-degree nodes are constructed in time $O(a + \log n)$, as every inactive node joins at most $d^*$ multicast groups, and the resulting trees have congestion $O(nd^*/n + \log n) = (a + \log n)$, w.h.p.
  Correspondingly, the multicast can be performed in time $O(a + \log n)$, w.h.p.

  We now bound the runtime of the final execution of the Identification Algorithm.
  First, note that the $s = \Theta(\log n)$ hash functions can be learned by broadcasting the $O(\log^2 n)$ bits required for each hash function (see Section~\ref{sec:primitives}) in a \emph{pipelined fashion} in a binary tree, which is implicitly given in the network.
  Clearly, this takes time $O(\log n)$ and requires each node to send and receive only $O(\log n)$ messages in each round.
  Every inactive node is a member of at most $O(a \log n)$ aggregation groups, and every node is a target of at most $4ec\log^2 n$ aggregation groups.
  By Lemma~\ref{lem:lowdeg} $\sum_{u \in U_{low}} (d(u) - d_i(u)) = O(an/\log n + n)$, w.h.p.
  As this is also a bound on the number of edges that participate in any trial, and each edge participates in $c \log n$ trials, the global load $L$ is bounded by $O(an + n\log n)$.
  Therefore, by Theorem~\ref{thm:aggregationProblem}, the Aggregation Algorithm takes time $O(a + \log n)$, w.h.p.
\end{proof}

The lemma below follows from the fact that $d_i^* = O(a)$.

\begin{lemma}\label{lem:stage3}
  Stage 3 takes time $O(a + \log n)$.
\end{lemma}

Finally, it remains to show that no node receives too many messages.

\begin{lemma}
  In each round of the algorithm, every node sends and receives at most $O(\log n)$ messages, w.h.p.
\end{lemma}

\begin{proof}
  By the discussion of Section~\ref{sec:primitives}, the executions of the Aggregation, Multicast Tree Setup, and Multicast Algorithm ensure that every node receives only $O(\log n )$ messages in each round.
  It remains to show the claim for the second step of Stage 2, where high-degree nodes broadcast their identifiers and receive their red edges, and for Stage 3, where active nodes learn which of their red edges lead to other active nodes.

  For the first part, note that after all high-degree nodes have broadcasted their identifiers, every active or waiting node sends out $O(\log n)$ messages containing its identifier in every round, w.h.p., which can easily be shown using Chernoff bounds.
  Second, as every high-degree node receives at most $d_i^*$ identifiers, it also follows from the Chernoff bound that every such node receives at most $O(\log n)$ messages in each round, w.h.p.

  Now consider Stage 3 of the algorithm.
  Again, by using the Chernoff bound, it can easily be shown that no node sends out more than $O(\log n)$ edge-messages in any round.
  Therefore, every node only receives $O(\log n)$ response messages in every round.
  It remains to show that every node receives at most $O(\log n)$ edge-messages in every round, from which it follows that it only sends out $O(\log n)$ response messages in every round.
  Let $A = \{ \{u,v\} \mid u \text{ or } v \text{ is active}\}$ and note that $|A| \le nd_i^*$.
  Fix a node $u \in V$ and a round $i \in \{1, \ldots, d_i^*\}$ and let $X_e$ be the binary random variable that is $1$ if and only if $h(\id(e)) = u$ and $r(\id(e)) = i$ for $e \in A$.
  Then $\Pr[X_e = 1] = 1/(nd_i^*)$.
  $X = \sum_{e \in A} X_e$ has expected value $E[X] \le 1$.
  Using the Chernoff bound we get that $X = O(\log n)$, w.h.p., which implies that $u$ receives at most $O(\log n)$ edge-messages in round $i$.
  The claim follows by taking the union bound over all nodes and rounds.
\end{proof}
} \fi
Taking Lemma~\ref{lem:orientationPhase} together with Lemma~\ref{lem:orientationHighLevel} yields the final theorem of this section.

\begin{theorem}\label{thm:orientation}
  The Orientation Algorithm computes an \iffull{}\fi 
  $O(a)$-orientation in time $O((a + \log n)\log n)$, w.h.p.
\end{theorem}

\section{Graph Problems Beyond MST}

We conclude our initiating study of the Node-Capacitated Clique by presenting a set of graph problems that can be solved efficiently in graphs with bounded arboricity.
The presented algorithms rely on a structure of precomputed multicast trees.
More specifically, for every node $u \in V$ we construct a multicast tree $T_{\id(u)}$ for the multicast group $A_{\id(u)} = N(u)$.
Since such trees enable the nodes to send messages to all of their neighbors, in the following we refer to them as \emph{broadcast trees}.

In a naive approach to construct these trees, one could simply use the Multicast Tree Setup Algorithm, where each node joins the multicast group of every neighbor.
However, as $\ell = \Delta$, the time to construct these trees would be $O(\overline{d} + \Delta/\log n + \log n)$, which can be $O(n/\log n)$ if $G$ is a star, for example.
Instead, we first construct an $O(a)$-orientation of the edges as shown in the previous section, and let $u$ only join multicast groups $A_{\id(v)}$ (which translates to injecting one packet per group into the butterfly) for every out-neighbor $v$.
Additionally, for every out-neighbor $v$ it takes care of $v$ joining $u$'s multicast group by injecting a packet for $v$.
In case of a star for example (whose arboricity is one), every node, including the center, injects at most two packets.
In general, we obtain the following result.

\begin{lemma}\label{lem:broadcasttrees}
  Setting up broadcast trees takes time $O(a + \log n)$, w.h.p.
  The congestion of the broadcast trees is $O(a + \log n)$, w.h.p.
\end{lemma}

The corollary below, which follows from the analysis of Theorem~
\ref{thm:multiaggregationProblem}, establishes one of the key techniques used by the algorithms in this section.

\begin{corollary}\label{cor:broadmultiaggregation}
  Let $S \subseteq V$.
  Using the broadcast trees, the Multi-Aggregation Algorithm solves any Multi-Aggregation Problem with multicast groups $A_{\id(u)} = N(u)$ and $s_{\id(u)} = u$ for all $u \in S$ in time $O(\sum_{u \in S}d(u)/n + \log n)$, w.h.p.
\end{corollary}

\subsection{Breadth-First Search Trees} \label{sec:bfs}

As a simple example, we show how to compute \emph{Breadth-First Search (BFS) Trees}:
Let $s$ be a node and let $\delta(u)$ be the length of a shortest (unweighted) path from $s$ to $u$ in $G$.
Furthermore, let $\pi(u)$ be the predecessor of $u$ on a shortest path from $s$ to $u$ (breaking ties by choosing the one with smallest identifier).
The goal is to let each node $u \in V$ eventually store $\delta(u)$ and $\pi(u)$.
Using the broadcast trees, the problem can easily be solved by the following algorithm, which proceeds in phases.
In Phase $1$, only $s$ is \emph{active}, and in Phase $i>1$, all nodes that have received an identifier in Phase $i-1$ for the first time are active.
In each phase, every active node sends its identifier to all of its neighbors using the broadcast trees and the Multi-Aggregation Algorithm.
By choosing $f$ as the minimum function, every node that has an active neighbor thereby receives the minimum identifier of all active neighbors.
Furthermore, in every Phase $i > 1$, every active node $u$ sets $\delta(u) = i-1$ and $\pi(u)$ to the node whose identifier it has received in the previous phase.
Clearly, after at most $D+1$ phases all nodes have been reached.

\begin{theorem}\label{thm:sssp}
	The algorithm computes a BFS Tree in time $O((a + D + \log n)\log n)$, w.h.p.
\end{theorem}

\begin{proof}
	By Lemma~\ref{lem:broadcasttrees}, the broadcast trees are constructed in time $O((a + \log n)\log n)$, w.h.p.
  Let $S_i$ be the set of nodes active in Phase $i$.
  By Corollary~\ref{cor:broadmultiaggregation}, the Multi-Aggregation Algorithm takes time $O(\sum_{u \in S_i}d(u)/n + \log n)$, w.h.p.
	We conclude a runtime of
	\begin{align*}
		&O\left((a + \log n)\log n + \sum_{i=1}^{D+1} \left(\sum_{u \in S_i}d(u) / n + \log n\right) \right) \\
    = \, &O\left((a + \log n)\log n + \sum_{u \in V}d(u)/n + (D+1)\log n\right)\\
		= \, &O((a + D + \log n)\log n), \text{ w.h.p.} \qedhere
	\end{align*}
\end{proof}

\subsection{Maximal Independent Set} \label{sec:mis}

In this section we show how to compute a \emph{maximal independent set (MIS)}:
A set $U \subseteq V$ is an MIS if (1) it is an independent set, i.e., no two nodes of $U$ are adjacent in $G$, and (2) there is no set $U' \subseteq U$ such that $U \subset U'$.
On a high level, we perform the algorithm of M\'{e}tivier et al~\cite{MRS11}, which works as follows.
First, all nodes are active and no node is in the MIS.
The algorithm proceeds in phases, where in each phase every active node $u$ first chooses a random number $r(u) \in [0, 1]$ and broadcasts the value to all of its neighbors.
$u$ then joins the MIS (and becomes inactive) if $r(u)$ is smaller than the minimum of all received values.
If so, it broadcasts a message to all of its neighbors, instructing them to become inactive.

We can easily perform a phase of the algorithm in our model by using two executions of the Multi-Aggregation Algorithm, the first to let every node aggregate the minimum of all values chosen by its neighbors, and the second to let every node that is not in the MIS determine whether it is adjacent to a node that is in the MIS.
This information is then used to determine whether the nodes have reached an MIS using the Aggregate-and-Broadcast Algorithm.
Since by \cite{MRS11} $O(\log n)$ phases suffice, and each phase can be performed in time $O(\overline{d} + \log n) = O(a + \log n)$ by Corollary~\ref{cor:broadmultiaggregation}, we conclude the following theorem.

\begin{theorem}
	The algorithm computes an MIS in time $O((a + \log n)\log n)$, w.h.p.
\end{theorem}


\subsection{Maximal Matching} \label{sec:matching}

Similar to an MIS, a \emph{maximal matching} $M \subseteq E$ is defined as a maximal set of independent (i.e., node-disjoint) edges.
To compute a maximal matching, we propose to use the algorithm of Israeli and Itai~\cite{II86}, which works as follows.
Initially, no node is matched.
The algorithm proceeds in phases, where in each phase every unmatched node $u$ performs the following procedure.
First, it chooses an edge to an unmatched neighbor uniformly at random.
If $u$ itself has been chosen by multiple neighbors, it accepts only one choice arbitrarily and informs the respective node.
The outcome is a collection of paths and cycles.
Each node of a path or cycle finally chooses one of its at most two neighbors.
If thereby two adjacent nodes choose the same edge, the edge joins the matching and the two nodes become matched.
Afterwards, all matched nodes and their incident edges are removed from the graph.

The algorithm lends itself to a realization using communication primitives.
First, we let every unmatched node randomly pick one of its unmatched neighbors by performing the Multi-Aggregation Algorithm with a slight modification.
Here, every node $u$ that is still unmatched multicasts a packet $p_{\id(u)}$ using its broadcast tree.
Recall that after $p_{\id(u)}$ has reached butterfly node $l(\id(u), v)$ for all $v \in N(u)$ in the execution of the Multi-Aggregation Algorithm, it is mapped to a new packet $(\id(v), p_{\id(u)})$.
Here, we additionally let $l(\id(u),v)$ choose a value $r \in [0,1]$ uniformly at random, and annotate $(\id(v), p_{\id(u)})$ by $r$.
Whenever thereafter two packets with the same target are combined, the packet annotated by the minimum value remains.
Thereby, every node that still has an unmatched neighbor receives the identifier of a node chosen uniformly and independently at random among its unmatched neighbors.

Afterwards, every node that has been chosen by multiple neighbors has to choose one of them arbitrarily.
This can be done by performing the Aggregation Algorithm, in which we let every node $u$ aggregate the minimum of the identifiers of all nodes by which it has been chosen in the previous step.
In the resulting collection of paths and cycles, neighbors can directly send messages to each other to determine which edges join the matching.
Finally, the nodes have to determine whether the matching is maximal, which can be done as described in the previous section.
Using Corollary 3.5 of \cite{II86} and Chernoff bounds, it can be shown that $O(\log n)$ phases suffice.
We conclude the following theorem.

\begin{theorem}
	The algorithm computes a maximal matching in time $O((a + \log n)\log n)$, w.h.p.
\end{theorem}

\subsection{$O(a)$-Coloring} \label{sec:coloring}

The goal of this section is to compute an \emph{$O(a)$-coloring}, in which every node has to choose one of $O(a)$ colors such that no color is chosen by two adjacent nodes.
Following the idea of Barenboim and Elkin~\cite{BE10}, we consider the partition of nodes into levels $L_1, \ldots, L_T$ and color the nodes of each level separately.
Recall that after the algorithm to compute the $O(a)$-orientation, every node knows the index of its own level.
Furthermore, for all $i$ every node $u \in L_i$ knows which of its neighbors are in \emph{lower levels} $L_1, \ldots, L_{i-1}$, the same level $L_i$, and \emph{higher levels} $L_{i+1}, \ldots, L_T$, since it knows which of its neighbors were inactive, active, or waiting in phase $i$.
First, the nodes use the Aggregate-and-Broadcast Algorithm to compute $\hat{a} = \max_{u \in V}\{\max(d_L(u), d_{out}(u))\} = O(a)$, where $d_L(u)$ is the number of neighbors of $u$ that are in the same level as $u$.
Furthermore, the nodes set up multicast trees for multicast groups $A_{\id(u)} = N_{in}(u)$ with source $s_{\id(u)} = u$ for all $u \in V$.
More precisely, every node joins the multicast group of each of its out-neighbors, which can be done in time $O(a + \log n)$, w.h.p., by Theorem~\ref{thm:treesetup}.

Afterwards, the algorithm proceeds in phases $1, \ldots, T$, where in each phase $i$ the nodes of level $L_{T - i + 1}$ get colored.
Throughout the algorithm's execution, every node $u$ maintains a color palette $C(u)$ initially set to $[2(1+\varepsilon)\hat{a}]$ for some constant $\varepsilon > 0$.
After each phase, the color palette of every remaining uncolored node has been narrowed down to all colors that have not yet been chosen by its neighbors.
Since every $u \in L_{T - i + 1}$ has at most $\hat{a}$ neighbors in higher levels, $C(u)$ still consists of at least $(1+\varepsilon)\hat{a}$ colors at the beginning of phase $i$.

In phase $i$ of the algorithm, the nodes of level $L_{T - i + 1}$ essentially perform the \emph{Color-Random Algorithm} of Kothapalli et al.~\cite{KOS+06}.
First, every node $u \in L_{T - i + 1}$ chooses a color $c_u$ from its color palette uniformly at random.
Then, it informs its in-neighbors about its choice by performing the Multicast Algorithm using the precomputed multicast trees and $\hat{a}$ as an upper bound on $\ell$.
Thereby, $u$ receives the colors chosen by its out-neighbors of the same level.
If $u$ does not receive its own color $c_u$, it permanently chooses $c_u$.
In that case, it first informs all of its in-neighbors about its permanent choice by again performing the Multicast Algorithm.
Afterwards, it informs all of its out-neighbors by performing the Aggregation Algorithm.
Here, $u$ is a member of aggregation groups $A_{\id(v) \circ c_u}$ for all $v \in N_{out}$ and target of aggregation groups $A_{\id(u) \circ i}$ for all $i \in [2(1+\varepsilon)\hat{a}]$.
Note that every node is a member of at most $\hat{a}$ and a target of at most $2(1+\varepsilon)\hat{a}$ aggregation groups.
Afterwards, all nodes (including nodes of lower levels) remove all colors permanently chosen by neighbors from their palettes.

The above procedure is repeated until all nodes of level $L_{T - i + 1}$ have permanently chosen a color, which is determined by performing the Aggregate-and-Broadcast Algorithm after each repetition.
Then, if $i > 1$, the nodes start the next phase, and terminate, otherwise.
The following theorem can be shown using the following facts:
(1) there are $O(\log n)$ phases, (2) $O(\sqrt{\log n})$ repetitions during a phase suffice until all nodes of the corresponding level are colored (see the discussion in Section 4 of \cite{KOS+06}), and (3) each repetition takes time $O(a + \log n)$.

\begin{theorem}\label{thm:coloringAlgorithm}
	The algorithm computes an $O(a)$-coloring in time $O((a + \log n)\log^{3/2} n)$, w.h.p.
\end{theorem}

\section{Conclusion} \label{sec:conclusion}

Our work initiates the study on the effect of node-capacities on the complexity of distributed graph computations.
We provide some ideas to approach the difficulties such limitations impose, which might be of interest for other problems as well.
Clearly, there is an abundance of classical problems that may be newly investigated under our model and for which our algorithms may be helpful.
In general, it would be interesting to see a classification of graph algorithms that can or cannot be efficiently performed in the Node-Capacitated Clique.
We are also very interested in proving lower bounds, which seems to be highly non-trivial in our model.
Particularly, we do not know whether the arboricity or the average node degree are natural lower bounds for some of the problems considered in this paper, although we highly suspect it.

Interestingly, the algorithms presented in this paper do not fully exploit the power of the Node-Capacitated Clique.
In fact, all of our algorithms still achieve the presented runtimes if in addition to knowing their neighbors in the input graph, they initially only know $\Theta(\log n)$ random nodes\footnote{Most communication in our algorithms is carried out using a butterfly as an overlay, which can be constructed, e.g., using \cite{AS18}.}.
It is an interesting question whether there are algorithms that actually require knowing all node identifiers.

\begin{acks}
J.A. is supported by DST/SERB Extra Mural Research Grant \iffull{}\fi
EMR/2016/003016, DST-DAAD Joint Project \iffull{\\}\fi
INT/FRG/DAAD/P-25/2018, and DST MATRICS MTR/2018/001198. 
M.G. is supported by SNSF Project No. 200021\_184735.
K.H. and C.S. are supported by DFG Project No. 160364472-SFB901 ("On-The-Fly-Computing").
F.K. is supported by ERC Grant No. 336495 (ACDC).
\end{acks}

\iffull{}\fi

%
\bibliographystyle{ACM-Reference-Format}
\bibliography{literature}

\iffull{

\clearpage

\appendix

\section{Simulations in the $k$-Machine Model} \label{app:simulation}

In this section we consider the simulation of an algorithm for the Node-Capacitated Clique in the $k$-machine model.
For the Congested Clique model, Klauck et al.~\cite{KNPR15} provide a conversion theorem that states the following.

\begin{theorem}[Theorem 4.1 in \cite{KNPR15}] \label{thm:conversion1}
  Any algorithm $A^C$ in the Congested Clique model that executes in $T^C$ rounds and passes at most $M^C$ messages over the course of the algorithm's execution can be simulated in the $k$-machine model so that it requires at most $\widetilde{O}(M^C/k^2 + T^C\Delta'/k)$ rounds.
  Here, $\Delta'$ is the communication degree complexity and refers to the maximum number of messages sent by any node at any round.
\end{theorem}

The simulation alluded to in Theorem~\ref{thm:conversion1} is quite straightforward.
Each node from the Congested Clique model is placed randomly on one of the $k$ machines in the $k$-machine model.
Under this random vertex partitioning scheme, each machine will get at most $\widetilde{O}(n/k)$ nodes from the Congested Clique model.
So it is natural for the messages sent by each node $u$ in the Congested Clique model to be simulated by the machine that holds $u$.

The following conversion result suited for the Node-Capacitated Clique model follows as a corollary when we notice that the number of messages per round is at most $\widetilde{O}(n)$ and, furthermore, $\Delta'$ under the Node-Capacitated Clique model is at most $O(\log n)$.

\begin{corollary} \label{cor:conversion2}
  Any algorithm $A^{NCC}$ in the Node-Capacitated Clique model that executes in $T^{NCC}$ rounds can be simulated in the $k$-machine model so that it requires at most $\widetilde{O}(nT^{NCC}/k^2 )$ rounds.
\end{corollary}

\section{Communication Primitives} \label{app:primitives}

In this section, we provide full descriptions of our communication primitives, and provide the missing proofs.
For simplicity, we refer to butterfly nodes as \emph{BF-nodes}.

\subsection{Aggregate-and-Broadcast Algorithm} \label{sec:aggregateAndBroadcast}

We first describe the Aggregate-and-Broadcast Algorithm of Theorem~\ref{thm:aggregateAndBroadcast} in detail.
First, every node that stores an input value, but does not emulate a node of the butterfly (in which case the most significant bit of its identifier must be $1$), sends it to the BF-node $j$ of level $0$ such that $j$ equals the remaining bits of its identifier.
Afterwards, every BF-node of level $0$ stores at most two input values, i.e., its own value and at most one value of a node that does not emulate a node of the butterfly.
Note, that for every BF-node of level $0$ there is a unique path of length $d$ from that node to any BF-node of level $d$ in the butterfly.
In the \emph{aggregation phase}, we send all input values to BF-node $0$ of level $d$, which in the following we refer to as the \emph{root} of the butterfly, along that path system.
Whenever two values $x, y$ reach the same BF-node $u$, $u$ only forwards $g(\{x,y\})$.
Thereby, the root eventually computes the aggregate of all values.
This value is finally broadcast to all BF-nodes of level $0$ in the \emph{broadcast phase}:
Every BF-node of level $i$ that receives the value forwards it to all of its neighbors in level $i-1$.
Finally, every node that does not emulate a BF-node gets informed by the BF-node of level $0$ whose identifier differs only in the most significant bit.
The correctness of Theorem~\ref{thm:aggregateAndBroadcast} can easily be seen.

In pointed out in the paper, we also use the above algorithm to achieve \emph{synchronization}:
Assume that the nodes execute some distributed algorithm that finishes in different rounds at the nodes.
In order to start a follow-up algorithm at the same round, the nodes can make use of the following slight modification of the Aggregate-and-Broadcast algorithm:
Every node delays its participation in the aggregation phase until it has finished the current algorithm.
Once it has finished, it sends a token to its corresponding BF-node at level $0$.
Once a BF-node at level $0$ has received a token from each of the at most two nodes of the Node-Capacitated Clique associated with it, it sends a token in the direction of the butterfly's root.
Similarly, once a BF-node at level $i>0$ has received tokens from both incoming edges, it sends a token in the direction of the root.
Thus, once the root has received tokens from both incoming edges, it knows that all nodes have finished the current algorithm.
The broadcast phase will then allow all nodes to start the follow-up algorithm at the same round.
It is easy to see that the synchronization just produces an overhead of $O(\log n)$ rounds.

\subsection{Aggregation Algorithm} \label{sec:aggregation}


Next, we describe the Aggregation Algorithm of Theorem~\ref{thm:aggregationProblem}.
We divide the execution of the algorithm into three phases, the \emph{Preprocessing Phase}, the \emph{Combining Phase}, and the \emph{Postprocessing Phase}.
First, in the \emph{Preprocessing Phase}, all input values are sent in batches of size $\lceil \log n \rceil$ to BF-nodes of level $0$ chosen uniformly at random.
More specifically, every node $u \in V$ transforms each input value $s_{u,i}$ for all $A_i$ of which $u$ is a member of into a packet of the form $(i, s_{u,i})$, and enumerates all of its packets arbitrarily from $1$ to $k \le \ell_1$ as $p_1, \ldots, p_k$.
Then, for each $j \in \{1,\ldots,\lceil k/\log n \rceil \}$, $u$ sends out packets $p_{(j-1)\lceil \log n \rceil + 1}, \ldots, p_{\min\{j \lceil \log n \rceil, k\}}$ in communication round $j$ to BF-nodes chosen uniformly and independently at random among all BF-nodes of level $0$.
To achieve synchronization after this phase, the nodes perform the Aggregate-and-Broadcast algorithm.

In the \emph{Combining Phase}, the input values of each aggregation group $A_i$ are aggregated to a node $h(i)$ (the \emph{intermediate target}) chosen uniformly and independently at random from the BF-nodes of level $d$ using a (pseudo-)random hash-function $h$.
This is achieved by using a variant of the \emph{random rank protocol} \cite{Ale82, Upf82}:
Each packet $p = (i, s_{u,i})$ stored at some BF-node of level $0$ gets assigned a $rank(p) = \rho(i)$ using some (pseudo-)random hash function $\rho: \{1, \ldots, N\} \rightarrow [K]$ that is known to all nodes.
Then, all packets belonging to aggregation group $A_i$ are routed towards their target $h(i)$ along the unique paths on the butterfly, and using the following rules:
\begin{enumerate}
  \item Whenever a BF-node stores multiple packets belonging to the same aggregation group $A_i$, it combines them into a single packet of rank $\rho(i)$, combining their values using the given aggregate function.
  \item Whenever multiple packets from different aggregation groups contend to use the same edge in the same round, the one with smallest rank wins (preferring the one with smallest aggregation group identifier in case of a tie), and all others get \emph{delayed}.
\end{enumerate}
Note that a packet can never get delayed by a packet belonging to the same aggregation group.
Clearly, in each round at most one packet is sent along each edge of the butterfly, and eventually all (combined) packets have reached their targets.

In order to determine whether the combining phase has finished, every BF-node of level $0$ sends out a token to all neighbors at level $1$ once it has sent out all packets.
Correspondingly, every BF-node at level $i > 0$ that has sent out all packets and has received tokens from both neighbors at level $i-1$ sends a token to both its neighbors at level $i+1$.
By performing the Aggregate-and-Broadcast Algorithm to determine whether all BF-nodes of level $d$ have received two tokens, the nodes eventually detect that the combining phase has finished.

Finally, in the \emph{Postprocessing Phase} the BF-nodes of level $d$ send their packets to the corresponding targets in rounds that are randomly chosen from $\{1, \ldots, s\}$, where $s = \lceil \hat{\ell_2}/\log n \rceil$.
More specifically, for each packet $p$ stored at some node $u$, which contains the result $f(\{s_{u,i} \mid u \in A_i\})$ for some aggregation group $A_i$, $u$ selects a round $r \in \{1, \ldots, s\}$ uniformly and independently at random and sends $p$ to $t_i$ in round $r$.
Again, the end of the phase is determined by using the Aggregate-and-Broadcast Algorithm.

We now turn to the analysis of the algorithm.

\begin{lemma}\label{lem:preprocessing}
  The Preprocessing Phase takes time $O(\ell_1/\log n)$.
  Moreover, in each round every node sends and receives at most $O(\log n)$ packets, w.h.p.
\end{lemma}

\begin{proof}
  The runtime and the bound on the number of packets sent out in each round are obvious.
  Hence, it remains to bound the number of packets that are received in each round.

  Fix any BF-node $u$ of level $0$ and round $t \in \{1, \ldots, \lceil \ell/\log n \rceil \}$.
  Altogether, at most $n \lceil \log n \rceil$ packets are sent out in round $t$, which we denote by $p_1, \ldots, p_{n \lceil \log n \rceil}$.
  For each $p_i$, let the binary random variable $X_i$ be $1$ if and only if $p_i$ is sent to BF-node $u$ in round $t$.
  Furthermore, let $X = \sum_{i=1}^k X_i$.
  Certainly, $\E[X_i] = \Pr[X_i=1] = 1/2^d$ and therefore, $\E[X] \le (n \lceil \log n \rceil)/2^d \le 2 \log n + 1$.
  Since the packets choose their destinations uniformly and independently at random, it follows from Lemma~\ref{lem:generalChernoffBound} that $X = O(\log n)$, w.h.p.
\end{proof}

In order to bound the runtime of the Combining Phase, we first analyze our variant of the random rank protocol in a general setting:
A path collection $P = \{p_1,\ldots,p_N\}$ in some graph $G$ is a \emph{leveled path collection} if every node $v$ can be given a level $l(v) \in \N$ so that for every edge $(v,w)$ of a path in that collection, $l(w)=l(v)+1$.
Given a leveled path collection $P$ of size $n$ in which packets belonging to the same aggregation group have the same destination, let the \emph{congestion} $C$ of $P$ be defined as the maximum number of aggregation groups that have packets that want to cross the same edge, and let the \emph{degree} $d$ of $P$ be defined as the maximum number of edges in $E(P)$ leading to the same node, where $E(P)$ is the set of all edges used by the paths in $P$.

\begin{theorem} \label{thm:combining1}
  For any leveled path collection $P$ of size $n$ with congestion $C$, depth $D$, and degree $d$, the routing strategy used in the Combining Phase with parameter $K \ge 8C$ needs at most $O(C + D \log d + \log n)$ steps, w.h.p., to finish routing in $P$.
\end{theorem}

\begin{proof}
  We closely follow the analysis of the random rank protocol in \cite{Sch98} and extend it with ideas from \cite{LMRR94} so that the analysis covers the case that packets can be combined.
  In order to bound the runtime, we will use the following delay sequence argument.

  Consider the runtime of the routing strategy to be at least $T\ge D+s$.
  We want to show that it is very improbable that $s$ is large.
  For this we need to find a structure that witnesses a large $s$.
  This structure should become more and more unlikely to exist the larger $s$ becomes.

  Let $p_1$ be a packet that arrived at its destination $v_1$ in step $T$, and let $A_1$ be the aggregation group of $p_1$.
  We follow the path of $p_1$ (or one of its predecessors, if $p_1$ is the result of the combination of two packets at some point) backwards until we reach a link $e_1$, where it was delayed the last time.
  Let us denote the length of the path from $v_1$ to $e_1$ (inclusive) by $l_1$, and the packet that delayed $p_1$ by $p_2$.
  Let $A_2$ be the aggregation group of $p_2$.
  From $e_1$ we follow the path of $p_2$ (or one of its predecessors) backwards until we reach a link $e_2$ where $p_2$ was delayed the last time, by a packet $p_3$ from some aggregation group $A_3$.
  Let us denote the length of the path from $e_1$ (exclusive) to $e_2$ (inclusive) by $l_2$.
  We repeat this construction until we arrive at a packet $p_{s+1}$ from some aggregation group $a_{s+1}$ that prevented the packet $p_s$ at edge $e_s$ from moving forward, and denote the number of links on the path of $p_i$ from $e_i$ (inclusive) to $e_{i-1}$ (exclusive) as $l_i$.
  Altogether it holds for all $i \in \{1, \ldots, s\}$:
  a packet from aggregation group $A_{i+1}$ is sent over $e_i$ at time step $T-\sum_{j=1}^i (l_j +1)+1$, and prevents at that time step a packet from aggregation group $A_i$ from moving forward.

  The path from $e_s$ to $v_1$ recorded by this process in reverse order is called {\em delay path}.
  It consists of $s$ contiguous parts of routing paths of length $l_1, \ldots, l_s \ge 0$ with $\sum_{i=1}^s l_i \le D$.
  Because of the contention resolution rule it holds $\rho(i) \ge \rho(i+1)$ for all $i \in \{1, \ldots, s\}$.
  A structure that contains all these features is defined as follows.

  \begin{definition}[{\boldmath $s$}-delay sequence]
    An {\em $s$-delay sequence} consists of \index{delay sequence}
    \begin{itemize}
      \item $s$ not necessarily different {\em delay links} $e_1,\ldots,e_s$;

      \item $s+1$ {\em delay groups} $a_1,\ldots,a_{s+1}$ such that the path of a packet from $a_i$ traverses $e_i$ and $e_{i-1}$ in that order for all $i\in\{2,\ldots,s\}$, the path of $p_1$ contains $e_1$, and the path of $p_{s+1}$ contains $e_s$;

      \item $s$ integers $l_1,\ldots,l_s\ge 0$ such that $l_1$ is the number of links on the path of $p_1$ from $e_1$ (inclusive) to its destination, and for all $i\in\{2,\ldots,s\}$, $l_i$ is the number of links on the path of $p_i$ from $e_i$ (inclusive) to $e_{i-1}$ (exclusive), and $\sum_{i=1}^s l_i \le D$; and

      \item $s+1$ integers $r_1,\ldots,r_{s+1}$ with $0\le r_{s+1} \le \ldots \le r_1<K$.
    \end{itemize}
    A delay sequence is called {\em active} if for all $i\in\{1,\ldots,s+1\}$ we have $\rho(a_i)=r_i$.
  \end{definition}

  \noindent
  Our observations above yield the following lemma.

  \begin{lemma}  \label{rr_le1}
    Any choice of the ranks that yields a routing time of $T\ge D+s$ steps implies an active $s$-delay sequence.
  \end{lemma}

  \begin{lemma} \label{rr_le2}
    The number of different $s$-delay sequences is at most
    \[
      n \cdot d^D \cdot C^s \cdot {D+s \choose s} \cdot {s+K \choose s+1}.
    \]
  \end{lemma}
  \begin{proof}
    There are at most ${D+s \choose s}$ possibilities to choose the $l_i$'s such that $\sum_{i=1}^s l_i \le D$.
    Furthermore, there are at most $n$ choices for $v_1$, which will also fix $a_1$.
    Once $v_1$ and $l_1$ is fixed, there are at most $d^{l_1}$ choices for $e_1$.
    Once $e_1$ is fixed, there are at most $d^{l_2}$ choices for $e_3$, and so on.
    So altogether, there are at most $d^D$ possibilities for $e_1,\ldots,e_s$.
    Since the congestion at every edge is at most $C$, there are at most $C$ possibilities for each $e_i$ to pick $a_{i+1}$, so altogether, there are at most $C^s$ possibilities to select $a_2,\ldots,a_{s+1}$.
    Finally, there are at most ${s+K \choose s+1}$ ways to select the $r_i$ such that $0\le r_{s+1} \le \ldots \le r_1<K$.
  \end{proof}

  Note that we assumed that there is a unique, total ordering on the ranks of the aggregation groups once $\rho$ is fixed.
  Hence, every aggregation group can only occur once in an $s$-delay sequence.
  Since $\rho$ is assumed to be a (pseudo-)random hash function, the probability that an $s$-delay sequence is active is $1/K^{s+1}$.
  Thus,
  \begin{align*}
    &\Pr[\text{The protocol needs at least $D+s$ steps}] \\
    \stackrel{\text{\scriptsize Lemma~\ref{rr_le1}}}{\le} &\Pr[\text{There exists an active $s$-delay sequence}] \\
    \stackrel{\text{\scriptsize Lemma~\ref{rr_le2}}}{\le} &n \cdot d^D \cdot C^s \cdot {D+s \choose s} \cdot {s+K \choose s+1} \cdot \frac{1}{K^{s+1}} \\
    \le \quad &n \cdot 2^{D \log d} \cdot C^s \cdot 2^{D+s} \cdot 2^{s+K} \cdot \frac{1}{K^{s+1}} \\
    \le \quad &n \cdot 2^{2s+D(\log d+1)+K} \cdot \left( \frac{C}{K} \right)^s.
  \end{align*}
  If we set $K\ge 8C$ and $s=K+D(\log d+1) + (\alpha+1)\log n$, where $\alpha >0$ is an arbitrary constant, then
  \begin{align*}
    &\Pr[\text{The algorithm needs at least $D+s$ steps}] \\
    \le &n \cdot 2^{2s+D(\log d +1)+K} \cdot 2^{-3s} \\
    = &n \cdot 2^{-s+D(\log d +1)+K} = \frac{1}{n^{\alpha}}
  \end{align*}
  which concludes the proof of Theorem~\ref{thm:combining1}.
\end{proof}

Using Theorem~\ref{thm:combining1}, we are now able to bound the runtime of the Combining Phase by determining the parameters of the underlying routing problem.

\begin{lemma}
  The Combining Phase takes time $O(L/n + \log n)$, w.h.p.
\end{lemma}

\begin{proof}
  The depth of the butterfly is $O(\log n)$ and its degree is $4$.
  Furthermore, the size of the routing problem is $L$.
  Therefore, it only remains to show that the congestion of the routing problem is $O(L/n+\log n)$, w.h.p.

  Consider some fixed edge $e$ from level $i$ to $i+1$ in the butterfly.
  For any $A \in \mathcal{A}$ let the binary random variable $X_A$ be $1$ if and only if there is at least one packet from $A$ crossing $e$.
  Clearly, there are $2^i \cdot 2^{d-i-1} = 2^d/2$ source-destination pairs, where the source is in level $0$ while the destination is in level $d$, whose unique shortest path passes through $e$.
  If the source of every packet is chosen uniformly and independently at random among all BF-nodes of level $0$ and the destinations of the aggregation groups are chosen uniformly and independently at random from all BF-nodes of level $d$, then the probability for an individual packet to pass through $e$ is $(2^d/2)/(2^d)^2 = 1/(2^{d+1})$.
  Hence, $\E[X_A] = \Pr[X_A=1] \le |A|/2^{d+1}$.
  Let $X=\sum_{A \in \mathcal{A}} X_A$.
  Then
  \[
    \E[X] = \sum_{A \in \mathcal{A}} \E[X_A] \le \frac{\sum_{A \in \mathcal{A}} |A|}{2^{d+1}} = \frac{L}{2^{d+1}} \le \frac{L}{n}.
  \]
  Since the $X_A$'s are independent, it follows from the Chernoff bounds (Lemma~\ref{lem:generalChernoffBound}) that $X = O(L/n + \log n)$, w.h.p.
\end{proof}

Using Chernoff bounds and the fact that every node at level $d$ of the butterfly is target of at most $O(\hat{\ell_2} + \log n)$ aggregation groups, w.h.p., the following result can be shown similarly to Lemma~\ref{lem:preprocessing}.

\begin{lemma}\label{lem:postprocessing}
  The Postprocessing Phase takes time $O(\hat{\ell_2} / \log n)$, w.h.p.
  Moreover, in each round every node sends and receives at most $O(\log n)$ packets, w.h.p.
\end{lemma}

We conclude the following theorem.

\begin{theorem}
  The Aggregation Algorithm takes time $O(L/n + (\ell_1 + \hat{\ell_2})/\log n + \log n)$, w.h.p.
\end{theorem}

\subsection{Multicast Tree Setup Algorithm} \label{sec:multicastTree}

First, every node $u$ injects an (empty) packet $(i, u)$ for each $i$ such that $u \in A_i$ into a BF-node $l(i, u)$ of level $0$ chosen uniformly and independently at random.
As in the Aggregation Algorithm, packets are sent in batches of size $\lceil \log n \rceil$.
Then, for all $i$, all packets of $A_i$ are aggregated at $h(i)$ using the same routing strategy as in the Aggregation Algorithm and an arbitrary aggregate function.
Alongside the algorithm's execution, every BF-node $u$ records for every $i \in \{1,\ldots,N\}$ all edges along which packets from group $A_i$ arrived during the routing towards $h(i)$, and declares them as edges of $T_i$.
Again, the intermediate steps are synchronized using the Aggregate-and-Broadcast Algorithm, and the final termination is determined using a token passing strategy.

The following theorem follows from the analysis of the Aggregation Algorithm.

\begin{theorem}
  The Multicast Tree Setup Algorithm computes multicast trees in time $O(L/n + \ell/\log n + \log n)$, w.h.p.
  The resulting multicast trees have congestion $O(L/n + \log n)$, w.h.p.
\end{theorem}

\subsection{Multicast Algorithm} \label{sec:multicast}

The Multicast Algorithm shares many similarities to the Aggregation Algorithm.
First, every source $s_i$ directly sends $p_i$ to $h(i)$.
Then, in the \emph{Spreading Phase}, $h(i)$ sends $p_i$ to all $l(i,u)$ for all $i$ and $u \in A_i$.
This is done by using the multicast trees and a variant our routing protocol of the Combining Phase:
First, each packet $p_i$ is assigned a $rank(p_i) = \rho(i)$.
Whenever a multicast packet $p_i$ of some aggregation group $A_i$ is stored by an inner node of $T_i$, i.e., by some BF-node $u$ of level $j \in \{1, \ldots, d\}$, then a copy of $p_i$ is sent over each outgoing edge of $u$ in $T_i$, i.e., towards one or both of $u$'s neighbors in level $j-1$.
If two packets from different multicast groups contend to use the same edge at the same time, the one with smallest rank is sent (preferring the one with smallest multicast group identifier in case of a tie), and the others get delayed.
Once there are no packets in transit anymore, which is determined by using the token passing strategy of the Aggregation Algorithm from level $0$ in the direction of level $d$, all leaves of the multicast trees have received their multicast packet.
Finally, every leaf node $l(i, u)$ sends $p_i$ to $u$ in a round randomly chosen from $\{1, \ldots, \lceil \hat{\ell}/\log n \rceil\}$.

The following theorem follows from the discussion of the previous sections and an adaptation of the delay sequence argument in the proof of Theorem~\ref{thm:combining1}.

\begin{theorem}
  The Multicast Algorithm takes time $O(C + \hat{\ell}/\log n + \log n)$, w.h.p.
\end{theorem}

\subsection{Multi-Aggregation Algorithm} \label{sec:multiAggregation}

The Multi-Aggregation Algorithm essentially first performs a multicast, then maps each multicast packet to a new aggregation group corresponding to its target, and finally aggregates the packets to their targets.
More precisely, first every node $s_i$ send its multicast packet to $h(i)$.
Then, by using the same strategy as in the Multicast Algorithm, we let each $l(i,u)$ receive $p_i$ for all $u \in A_i$ and all $i$.
Every node $l(i,u)$ then \emph{maps} $p_i$ to a packet $(\id(u), p_i)$ for all $u \in A_i$ and all $i$.
We randomly distribute the resulting packets by letting each BF-node send out its packets, one after the other, to BF-nodes of level $0$ chosen uniformly and independently at random.
By using the same strategy as in the Aggregation Algorithm, we then aggregate all packets $(\id(u), p_i)$ for all $i$ to $h(\id(u))$, and finally send the result $f(\{p_i \mid u \in A_i\})$ from $h(\id(u))$ to $u$.

The following theorem follows from discussion of the previous sections and from the fact that the mapping takes time $O(C)$.

\begin{theorem}\label{thm:multiAggregation}
  The Multi-Aggregation Algorithm takes time $O(C + \log n)$, w.h.p.
\end{theorem}

} \fi

\end{document}